%% file: paper.tex
\def\arXiv{}

\ifdefined \arXiv
	\documentclass[letter,11pt]{article}
\else
	\documentclass{jocg}
\fi

\usepackage{acronym}
\usepackage[boxed,linesnumbered,vlined]{algorithm2e}
\usepackage{amsfonts}
\usepackage{amsmath}
\usepackage{amssymb}
\usepackage{amsthm}
\usepackage[american]{babel}
\usepackage{dsfont}
\usepackage{floatrow}
\usepackage{graphicx}
\usepackage[mathlines]{lineno}
\usepackage{multirow}
\usepackage{paralist}
\usepackage{subfigure}
\usepackage[disable]{todonotes}
\usepackage{xspace}

\ifdefined \arXiv
	\usepackage[top=1.25in,left=1.25in,bottom=1.25in,right=1.25in]{geometry}
	\usepackage{url}
	\usepackage[pdfborder={0 0 0}]{hyperref}
\fi

\acrodef{AGP}{Art Gallery Problem}
\acrodef{CGAL}{Computational Geometry Algorithms Library}
\acrodef{CTGP}{Continuous Terrain Guarding Problem}
\acrodef{EGC}{Exact Geometric Computation}
\acrodef{IP}{Integer Linear Program}
\acrodef{PTAS}{Polynomial Time Approximation Scheme}
\acrodefplural{PTAS}[PTAS's]{Polynomial Time Approximation Schemes}
\acrodef{SC}{Set Cover}
\acrodef{TGP}{Terrain Guarding Problem}
\acrodef{TGPIL}{Terrain Guarding Problem Instance Library}
\acrodef{VTGP}{Terrain Guarding Problem with Vertex Guards}

\SetKwComment{Comment}{$\triangleright$\ }{}
\SetCommentSty{AlFnt}

\SetNlSty{small}{}{:}
\DontPrintSemicolon
\SetKwInOut{Input}{input}
\SetKwInOut{Output}{output}

\newtheorem{definition}{Definition}[section]
\newtheorem{theorem}[definition]{Theorem}
\newtheorem{lemma}[definition]{Lemma}
\newtheorem{corollary}[definition]{Corollary}
\newtheorem{observation}[definition]{Observation}

\ifdefined \arXiv
	\nolinenumbers
\else
	\linenumbers
	\nolinenumbers
\fi

\setdefaultenum{(1)}{(a)}{(i)}{}

\def\dash---{\kern.16667em---\penalty\exhyphenpenalty\hskip.16667em\relax}
\newcommand{\bigo}{\operatorname{o}}
\newcommand{\bigO}{\operatorname{O}}

\newcommand{\bigTheta}{\operatorname{\Theta}}
\newcommand{\interior}{\operatorname{int}}
\newcommand{\N}{\mathds{N}}

\newcommand{\opt}{\operatorname{OPT}}
\newcommand{\Q}{\mathds{Q}}
\newcommand{\R}{\mathds{R}}
\newcommand{\tgp}{\operatorname{TGP}}
\newcommand{\V}{\operatorname{\mathcal{V}}}
\newcommand{\VK}{\operatorname{VK}}

\newcommand{\pointguardmode}{\textsc{Point\-Guards}\xspace}
\newcommand{\vertexguardmode}{\textsc{Ver\-tex\-Guards}\xspace}
\newcommand{\domfilter}{\textsc{Dom\-i\-na\-tion\-Fil\-ter}\xspace}
\newcommand{\pointguardfilter}{\textsc{Edge\-Fil\-ter}\xspace}
\newcommand{\witnessfilter}{\textsc{Wit\-ness\-Fil\-ter}\xspace}

\newcommand{\walk}{\mbox{\textsc{Walk}}\xspace}
\newcommand{\sinewalk}{\textsc{Sine\-Walk}\xspace}
\newcommand{\parabolawalk}{\textsc{Pa\-rab\-o\-la\-Walk}\xspace}
\newcommand{\concavevalleys}{\textsc{Con\-cave\-Val\-leys}\xspace}

\newcommand{\vdefault}{\textsc{VDe\-fault}\xspace}
\newcommand{\vnodom}{\textsc{VNo\-Dom}\xspace}
\newcommand{\vnow}{\mbox{\textsc{VNoW}}\xspace}
\newcommand{\pdefault}{\textsc{PDe\-fault}\xspace}
\newcommand{\pnoedge}{\textsc{PNo\-Edge}\xspace}
\newcommand{\pnodom}{\textsc{PNo\-Dom}\xspace}
\newcommand{\pnow}{\mbox{\textsc{PNoW}}\xspace}

\newcommand{\papertitle}{The Continuous 1.5D Terrain Guarding Problem:\protect\\ Discretization, Optimal Solutions, and {PTAS}}
\newcommand{\paperthanks}{This work extends and subsumes Chapter~3 of \emph{James King's PhD thesis,} pages 29--72, 2010~\cite{k-gpgst-10}, and the extended abstracts that appeared in the \emph{Proceedings of the 26th Canadian Conference on Computational Geometry (CCCG~2014),} pages 367--373, 2014~\cite{fhs-ptasctgp-14} and in the \emph{31st European Workshop on Computational Geometry (EuroCG~2015),} pages 212--215, 2015~\cite{fhs-esctg-15}.}

\ifdefined \arXiv
	\title{\bf\Large\papertitle\thanks{\paperthanks}}
	\date{}
\else
	\title{\MakeUppercase{\papertitle}\thanks{\paperthanks}}
\fi

\ifdefined \arXiv
	\newcommand{\affil}[1]{#1}
	\newcommand{\email}[1]{\texttt{#1}}
\fi

\author{%
	Stephan~Friedrichs,%
	\thanks{%
		\affil{Max Planck Institute for Informatics, Saarbr\"ucken, Germany},
		\email{sfriedri@mpi-inf.mpg.de}}
	\thanks{\affil{Saarbr\"ucken Graduate School of Computer Science}}\,
	Michael~Hemmer,%
	\thanks{%
		\affil{TU Braunschweig, IBR, Algorithms Group, Braunschweig, Germany},
		\email{mhsaar@gmail.com}}\,
	James~King,%
	\thanks{%
		\affil{D-Wave Systems, Burnaby, Canada},
		\email{jking@dwavesys.com}}~
	and Christiane~Schmidt%
	\thanks{%
		\affil{Communications and Transport Systems, ITN, Link\"oping University, Sweden. Supported by grant 2014-03476 from Sweden's innovation agency VINNOVA.}
		\email{christiane.schmidt@liu.se}}
}

\begin{document}

\maketitle

\input{abstract.tex}
\acresetall
\input{introduction.tex}
\acresetall
\input{discretization.tex}
\acresetall
\input{complexity.tex}
\acresetall
\input{ptas.tex}
\acresetall
\input{filters.tex}
\acresetall
\input{algorithm.tex}
\acresetall
\input{experiments.tex}
\acresetall
\input{conclusion.tex}

\bibliographystyle{plain}
\bibliography{bibliography}

\end{document}

%% file: abstract.tex
\begin{abstract}
In the NP-hard~\cite{kk-tginph-11} continuous 1.5D \ac{TGP} we are given an $x$-monotone chain of line segments in $\R^2$ (the \emph{terrain~$T$}) and ask for the minimum number of guards (located anywhere on~$T$) required to guard all of~$T$.
We construct guard candidate and witness sets $G,W \subset T$ of polynomial size such that any feasible (optimal) guard cover $G^* \subseteq G$ for $W$ is also feasible (optimal) for the continuous \ac{TGP}.
This discretization allows us to
\begin{inparaenum}
\item
	settle NP-completeness for the continuous \ac{TGP},
\item
	provide a \ac{PTAS} for the continuous \ac{TGP} using the \ac{PTAS} for the discrete \ac{TGP} by Gibson et~al.~\cite{gkkv-gtvls-14}, and
\item
	formulate the continuous \ac{TGP} as an \ac{IP}.
\end{inparaenum}
Furthermore, we propose several filtering techniques reducing the size of our discretization, allowing us to devise an efficient \acs{IP}-based algorithm that reliably provides optimal guard placements for terrains with up to $10^6$ vertices within minutes on a standard desktop computer.
\end{abstract}

%% file: introduction.tex
\section{Introduction}
\label{sec:introduction}

In the 1.5D \ac{TGP}, we are given an $x$-monotone chain of line segments in~$\R^2$, the terrain~$T$, on which we have to place a minimum number of point-shaped guards, such that they cover~$T$.
This is a close relative of the \ac{AGP} and traditionally motivated by the optimal placement of antennas for line-of-sight communication networks, or the placement of street lights or security cameras along roads~\cite{bkm-acfaafotg-07}.

The authors would like to revive a motivation stemming from research regarding algorithms solving the \ac{AGP}~\cite{rsfhkt-eag-14,eh-ggt-06,efhkkms-aagi-14,ffks-ffagp-15,kbfs-esbgagp-12} already mentioned in~\cite{bkm-acfaafotg-07}:
An application of the \ac{AGP} is the placement of sensors or communication devices w.r.t.\ obstacles, for example placing laser scanners in production facilities to acquire a precise mapping of the facility~\cite{efhkkms-aagi-14,kbfs-esbgagp-12}.
While the \ac{AGP} properly models most indoor environments it cannot capture many outdoor scenarios, like placing cell phone towers in an urban environment, because it does not take height information into account.
To remedy this shortcoming essentially means working on two dimensions and height, a 2.5D \ac{AGP}.
One dimension and height, the 1.5D \ac{TGP}, is a natural starting point to develop techniques for a 2.5D \ac{AGP}.
We show in this paper that the ``height dimension'' is more benevolent than the ``second dimension'' in the \ac{AGP}:
It allows a finite discretization whose existence in the \ac{AGP} is, to the best of our knowledge, still unknown and poses a key challenge w.r.t.\ software solving the \ac{AGP}~\cite{rsfhkt-eag-14}.
We hope that our contribution helps tackling the 2.5D \ac{AGP}.

\subsection{Our Contribution}

\begin{enumerate}
\item
	Our core contribution is to show that the \ac{CTGP}, where guards can be freely placed on the terrain, has a discretization of polynomial size (Section~\ref{sec:discretization}).
	We then infer two results:
	\begin{enumerate}
	\item
		While the \ac{CTGP} is known to be NP-hard~\cite{kk-tginph-11}, we also conclude that it is a member of NP, and hence NP-complete (Section~\ref{sec:complexity}).

	\item
		It follows from the \ac{PTAS} for the discrete \ac{TGP} from Gibson et~al.~\cite{gkkv-gtvls-14} that there is a \ac{PTAS} for the \ac{CTGP} (Section~\ref{sec:ptas}).
	\end{enumerate}

\item
	We present filtering techniques reducing the size of our discretization (Section~\ref{sec:filters}).

\item
	An efficient algorithm for continuous and discrete \ac{TGP} versions is proposed.
	It finds optimal solutions for terrains with up to~$10^6$ vertices on a standard desktop computer\footnote{Standard as of~2015: An Intel Core i7-3770 CPU with 3.4\,GHz with 14\,GB of main memory.} within minutes.
	This is achieved following the \ac{EGC} paradigm, i.e., using exact arithmetic for geometric calculations to ensure correctness.
	We test our algorithm and filtering techniques (Sections~\ref{sec:implementation} and~\ref{sec:experiments}).
\end{enumerate}

\begin{figure}
	\centering
	\subfigure[The visibility region $\V(p)$ of $p \in T$ has $\bigO(n)$ subterrains.]{
		\includegraphics[width=0.61\textwidth]{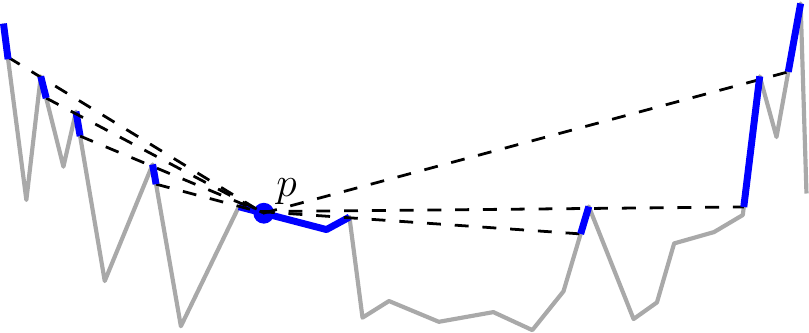}
		\label{fig:tgp-visibility}
	}
	\subfigure[This terrain needs two vertex- but only one non-vertex guard~\cite{bkm-acfaafotg-07}.]{
		\includegraphics[width=0.61\textwidth]{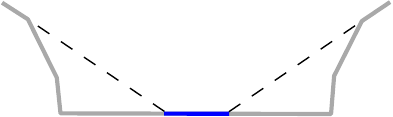}
		\label{fig:tgp-point}
	}
	\caption{The \acf{TGP}: visibility~\subref{fig:tgp-visibility} and non-vertex guards~\subref{fig:tgp-point}.}
	\label{fig:tgp}
\end{figure}

\subsection{Related Work}

The \ac{TGP} is closely related to the \ac{AGP} where, given a polygon~$P$, we seek a minimum cardinality guard set that covers~$P$.
Potential guards can, e.g.,  be located on the vertices only, on arbitrary points in~$P$, or patrol along edges or diagonals of~$P$.
Many polygon classes have been considered for the \ac{AGP}, including simple polygons, polygons with holes, and orthogonal polygons.
Moreover, the guards' task can be altered, e.g.\ Laurentini~\cite{l-gwag-99} required visibility coverage for the edges of~$P$, but not the interior.

The first result in the context of the \ac{AGP} was obtained by Chv\'{a}tal~\cite{c-actpg-75} who proved the \emph{Art Gallery Theorem,} answering a question posed by Victor Klee in~1973 (see~\cite{r-agta-87}):
$\lfloor \frac{n}{3} \rfloor$~guards are always sufficient and sometimes necessary to guard a polygon of $n$ vertices.
A simple and elegant proof of the sufficiency was later given by Fisk~\cite{f-spcwt-78}.
Related results were obtained for various polygon classes, Kahn et~al.~\cite{kkk-tgrfw-83} established a tight bound of $\lfloor \frac{n}{4} \rfloor$ for orthogonal polygons with $n$ vertices.

The work of Chv\'{a}tal and its variants focused on upper bounds on the number of guards.
However, the \ac{AGP} also is an optimization problem (given a polygon, find a minimum number of guards covering it), the decision variant of which was shown to be NP-hard for various problem versions~\cite{rs-snphpdp-83,sh-tnphagpop-95}, even for vertex guards in polygons without holes~\cite{ll-ccagp-86}.
Eidenbenz et~al.\ established APX-hardness of many \ac{AGP} variants~\cite{esw-irgpt-01}.
Chwa et~al.~\cite{cjkmos-gagbgw-06} considered witnessable polygons, in which coverage of some finite set of witness points implies coverage of the entire polygon.
For classical surveys on the \ac{AGP} see O'Rourke~\cite{r-agta-87} and Shermer~\cite{s-rrag-92}, and de Rezende et~al.~\cite{rsfhkt-eag-14} for more recent computational developments.

For the~1.5D \ac{TGP}, research first focused on approximation algorithms, because NP-hardness was generally assumed, but had not been established.
The first constant-factor approximation was given by Ben-Moshe et~al.~\cite{bkm-acfaafotg-07} for the discrete vertex guard problem version $\tgp(V,V)$,\footnote{$\tgp(G,W)$ means that $W$ must be covered using only guards located in~$G$, see Definition~\ref{def:tgp}.} where only vertex guards are used to cover only the vertices.
They were able to use it as a building block for an $\bigO(1)$-approximation of $\tgp(T,T)$, where guards on arbitrary locations on $T$ must guard all of~$T$.
The approximation factor of this algorithm was not stated by the authors, but claimed to be 6 in~\cite{k-4aagdt-06} (with minor modifications).
Another constant-factor approximation based on $\epsilon$-nets and \ac{SC} was given by Clarkson and Varadarajan~\cite{cv-iaags-07}.
King~\cite{k-4aagdt-06} presented a 4-approximation (which was later shown to actually be a 5-approximation~\cite{k-eaag}) for $\tgp(V,V)$ and $\tgp(T,T)$.
The most recent $\bigO(1)$-approximation was presented by Elbassioni et~al.~\cite{ekmms-iag15-11}:
Using LP-rounding techniques, they achieve a 4-approximation of $\tgp(T,T)$ and $\tgp(G,W)$ w.r.t.\ finite, disjoint $G, W \subset T$ (a~5-approximation if $G \cap W \neq \emptyset$).
This approximation is also applicable to the \ac{TGP} with weighted guards.
In the~2009 conference version of~\cite{gkkv-gtvls-14}, Gibson et~al.\ devised a \ac{PTAS} based on local search for $\tgp(G,W)$ and $\tgp(G,T)$, where $G,W \subset T$ are finite.

Only after all these approximation results, in the 2010 conference version of~\cite{kk-tginph-11}, King and Krohn established the NP-hardness of both the discrete and the continuous \ac{TGP} by a reduction from {PLANAR~3SAT}.
The membership of the \ac{CTGP} in NP remained, to the best of our knowledge, an open problem that we answer positively in Section~\ref{sec:complexity}.
Khodakarami et~al.~\cite{kdm-fpagt-15} showed that the \ac{TGP} is fixed-parameter tractable w.r.t.\ the \emph{depth of terrain onion peeling,} the number of layers of upper convex hulls induced by a terrain.

Variants of the \ac{TGP} include guards hovering above the terrain (Eidenbenz~\cite{e-aatg-02}), orthogonal terrains (Katz and Roisman~\cite{kr-gvrd-08}), and directed visibility (Durocher et~al.~\cite{dlm-got-15}).
Hurtado et~al.~\cite{hlmssss-tvwmv-14} gave algorithms for computing visibility regions in~1.5D and~2.5D terrains.
Haas and Hemmer~\cite{hh-tv-15} presented implementations for~1.5D visibility based on~\cite{hlmssss-tvwmv-14} and the triangular expansion technique for visibility computations in polygons by Bungiu et~al.~\cite{bhhhk-ecvp-14}.

Martinovi{\'c} et~al.~\cite{mms-epiaagt-15} proposed an approximate solver for the discrete \ac{TGP}.
Requiring a-priori knowledge about pairwise visibility of the vertices~$V$, they 5.5- and 6-approximate $\tgp(V,V)$ instances with up to 8000 vertices and dense (0.19--0.65) visibility matrices in 11--900 and 4--250 seconds.
As geometric information is encoded in the input, they are not tied to the \ac{EGC} paradigm, use floating-point arithmetic, and a parallel GPU implementation.
Note that we solve a different problem:
We determine the discretization and visibility information that Martinovi{\'c} et~al.\ require as input, follow the \ac{EGC} paradigm, and guarantee optimal solutions in no more than~3.5 seconds for~10000 vertices.
However, we do not focus on dense visibility matrices and use different hardware, rendering a comparison of computation times meaningless.

Regarding a discretization for the continuous \ac{TGP}, Gibson et~al.\ claimed~\cite{gkkv-gtvls-14} that their local search works well, but could not limit the number of bits representing the guards.
King gave a discretization with $\bigO(n^3)$ guard candidates and $\bigO(n^4)$ witnesses in his PhD thesis in 2010~\cite{k-gpgst-10} and posed the question if a smaller discretization exists.
Independently, Friedrichs et~al.\ discovered a discretization using $\bigO(n^2)$ guard candidates and $\bigO(n^3)$ witnesses~\cite{fhs-ptasctgp-14} in~2014.
This paper subsumes and extends~\cite{fhs-ptasctgp-14,k-gpgst-10}.

\subsection{Preliminaries and Notation}

A \emph{terrain~$T$,} see Figure~\ref{fig:tgp}, is an $x$-monotone chain of line segments in $\R^2$ defined by its \emph{vertices} $V(T) = \{ v_1, \dots, v_n \}$ that has \emph{edges} $E(T) = \{ e_1, \dots, e_{n-1} \}$ with $e_i = \overline{v_i v_{i+1}}$.
Unless specified otherwise, $n := |V(T)|$.
Where $T$ is clear from context, we occasionally abbreviate $V(T)$ and $E(T)$ by $V$ and~$E$.
$v_i$~and $v_{i+1}$ are the vertices of the edge~$e_i$, and $\interior(e_i) := e_i \setminus \{v_i, v_{i+1}\}$ is its \emph{interior.}
Due to monotonicity, the points on $T$ are totally ordered w.r.t.\ their $x$-coordinates.
For $p,q \in T$, we write $p \leq q$ ($p < q$) if $p$ is (strictly) left of~$q$, i.e., has a (strictly) smaller $x$-coordinate.
We refer to a closed, connected subset of $T$ as a \emph{subterrain.}

A point $p \in T$ \emph{sees} or \emph{covers} $q \in T$ if and only if $\overline{pq}$ is nowhere below~$T$.
$\V(p)$~is the \emph{visibility region} of $p$ with $\V(p) := \{ q \in T \mid \textnormal{$p$ sees $q$} \}$.
Observe that $\V(p)$ is not necessarily connected and is the union of $\bigO(n)$ subterrains, see Figure~\ref{fig:tgp-visibility}.
We say that $q \in \V(p)$ is \emph{extremal} in $\V(p)$ if $q$ has a maximal or minimal $x$-coordinate within its subterrain in~$\V(p)$.
For $G \subseteq T$ we abbreviate $\V(G) := \bigcup_{g \in G} \V(g)$.
A set $G \subseteq T$ with $\V(G) = T$ is named a \emph{(guard) cover} of~$T$.
In this context, $g \in G$ is sometimes referred to as \emph{guard.}

\begin{definition}[\acl{TGP}]\label{def:tgp}
	In the \emph{\acf{TGP},} abbreviated $\tgp(G,W)$, we are given a terrain~$T$, and sets of guard candidates and witnesses $G, W \subseteq T$.
	$C \subseteq G$ is \emph{feasible w.r.t.\ $\tgp(G,W)$} if and only if $W \subseteq \V(C)$.
	If $C$ is feasible and $|C| = \opt(G,W) := \min\{ |C| \mid \text{$C \subseteq G$ is feasible w.r.t.\ $\tgp(G,W)$} \}$, we say that $C$ is \emph{optimal w.r.t.\ $\tgp(G,W)$.}
	$\tgp(G,W)$ asks for an optimal guard cover $C \subseteq G$.
	The \emph{\acf{CTGP}} is $\tgp(T,T)$, and the \emph{\acf{VTGP}} is $\tgp(V(T),T)$.
\end{definition}

Throughout this paper, we assume $W \subseteq \V(G)$, i.e., that $\tgp(G,W)$ has a feasible solution.
The \ac{CTGP} is the primary focus of this paper.
Observe that \ac{CTGP} and \ac{VTGP} are different problems~\cite{bkm-acfaafotg-07}, as demonstrated in Figure~\ref{fig:tgp-point}.
We consider \ac{VTGP} a representative of the numerous discrete versions of the \ac{TGP};
our algorithm solves both \ac{CTGP} and \ac{VTGP}, and generalizes to arbitrary discretizations.

%% file: discretization.tex
\section{Discretization}
\label{sec:discretization}

This section is our core contribution.
We consider the following problem:
Given a terrain $T$ with $n$ vertices, construct sets $G, W \subset T$ (guard candidate and witness points) of size polynomial in~$n$, such that any feasible (optimal) solution for $\tgp(G,W)$ is feasible (optimal) for $\tgp(T,T)$ as well.
We proceed in three steps.
\begin{inparaenum}
\item
	In Section~\ref{sec:discretization-witnesses} we assume that we are provided with some finite guard candidate set $G \subset T$ and show how to construct a witness set $W(G)$ with $|W(G)| \in \bigO(n|G|)$, such that any feasible solution of $\tgp(G, W(G))$ is feasible for $\tgp(G,T)$ as well.
\item
	Section~\ref{sec:discretization-guards} discusses a set of guard candidates $U$ with $|U| \in \bigO(n^2)$ and $\opt(U,T) = \opt(T,T)$.
\item
	We combine the above steps in Section~\ref{sec:discretization-complete}.
\end{inparaenum}

\begin{figure}
	\centering
	\includegraphics[width=0.75\textwidth]{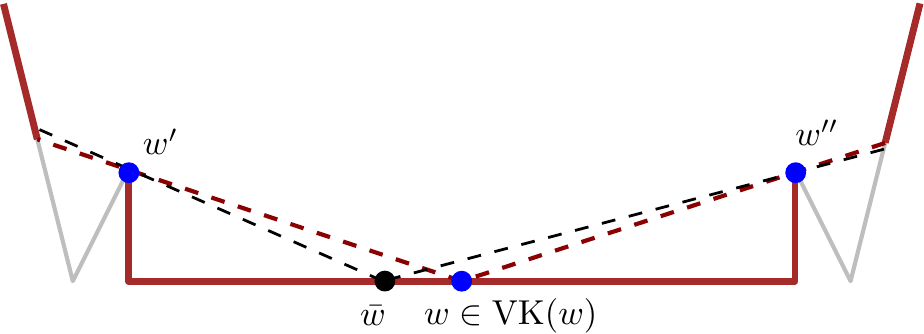}
	\caption{Witness $w$, $\V(w)$ highlighted in red, and its finite visibility kernel $\VK(w) = \{w, w', w''\}$ marked in blue.
		$\bar w$ has equivalent properties.}
	\label{fig:vis-kernel}
\end{figure}

When discretizing a problem as closely related to the \ac{AGP} as the \ac{TGP}, one must consider the work of Chwa et~al.\ who pursued the idea of \emph{witnessable} polygons~\cite{cjkmos-gagbgw-06} which allow placing a finite set of witnesses, such that covering the witnesses with any guard set implies full coverage of the polygon.
The basic building blocks of Chwa et~al.\ are \emph{visibility kernels:}
Given a point $w$ in a polygon, the visibility kernel of $w$ is the set of points that see at least as much as $w$ (definition for terrains below).
Chwa et~al.\ show that a polygon admits a finite witness set if and only if it can be covered by a finite set of visibility kernels;
this is not the case for arbitrary polygons.

Transferring this approach to the \ac{TGP} means that the visibility kernel of $w \in T$ is $\VK(w) := \{ w' \in T \mid \V(w) \subseteq \V(w') \}$.
Then for the terrain $T$ and $w \in T$ in Figure~\ref{fig:vis-kernel} we have $\VK(w) = \{ w, w', w'' \}$, so $\VK(w)$ is finite.
The same argument holds for infinitely many $\bar{w} \in T$ near~$w$.
It follows that $T$ does not admit a finite visibility kernel cover and thus is not witnessable as defined by Chwa et~al.\ (our witness set below is different in that it is associated with a finite guard set).

\subsection{Witnesses}
\label{sec:discretization-witnesses}

\begin{figure}
	\centering
	\subfigure[%
			Visibility overlay of $\V(g_1)$, $\V(g_2)$, and $\V(g_3)$ indicated in blue, orange, and green, respectively.
			Overlaps are indicated by altering colors.]{
		\includegraphics[width=.65\linewidth]{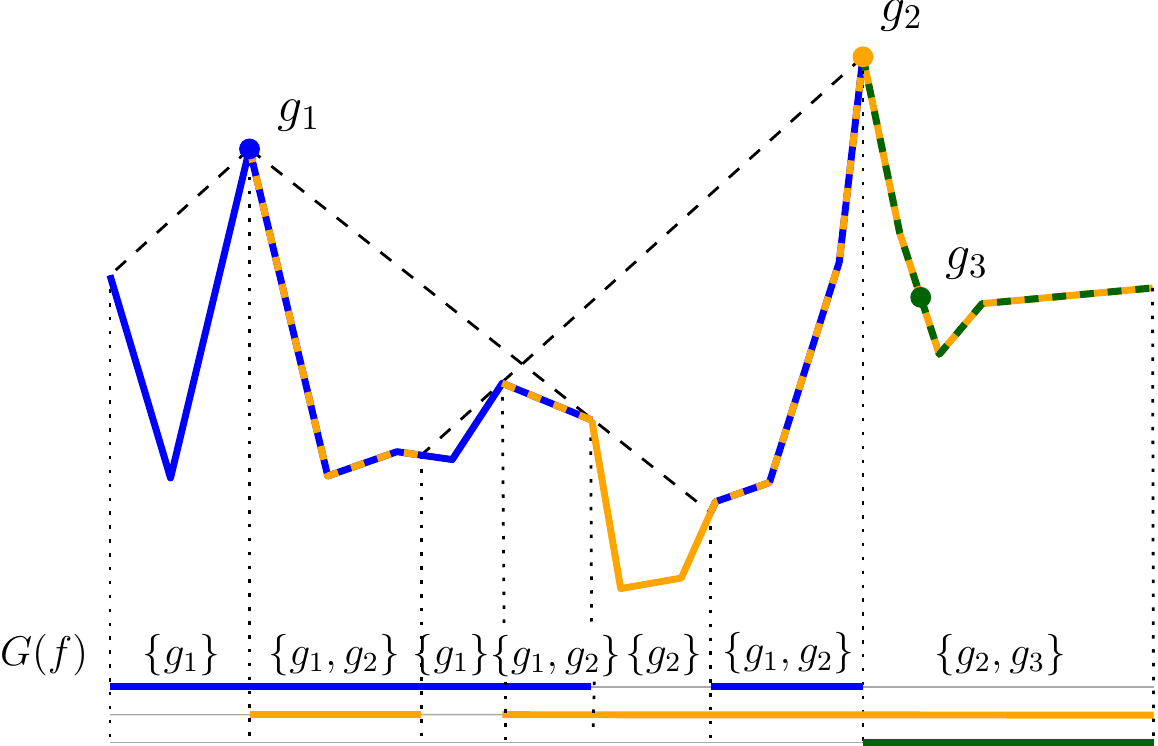}
		\label{fig:witnesses-overlay}
	}
	\subfigure[The set of inclusion-minimal features may still have cardinality $\bigO(n|G|)$.]{
		\includegraphics[width=.9\textwidth]{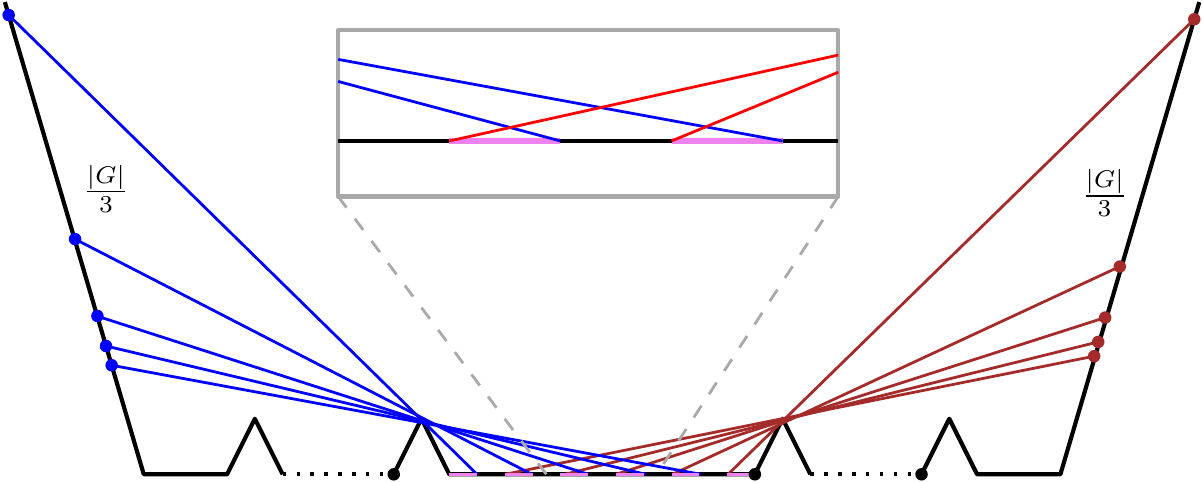}
		\label{fig:witnesses-inclusion-minimal}
	}
	\caption{Witness discretization: visibility overlay~\subref{fig:witnesses-overlay} and cardinality~\subref{fig:witnesses-inclusion-minimal}.}
	\label{fig:witnesses}
\end{figure}

Suppose we are given a terrain $T$ and a finite set $G \subset T$ of guard candidates with $\V(G) = T$, and we want to cover $T$ using only guards $C \subseteq G$, i.e., we want to solve $\tgp(G,T)$.
$G$~could be the set $V(T)$ of vertices to solve the \ac{VTGP} or any other finite set, especially our guard candidates in Equation~\eqref{eq:u}.
We construct a finite set $W(G) \subset T$ of $\bigO(n|G|)$ witness points, such that feasible solutions for $\tgp(G,W(G))$ also are feasible for $\tgp(G,T)$.

Let $g \in G$ be one of the guard candidates.
$\V(g)$~subdivides $T$ into $\bigO(n)$ subterrains, see Figure~\ref{fig:tgp-visibility}.
The monotonicity of $T$ allows us to project them onto the $x$-axis and thus to represent $\V(g)$ as a set of closed \emph{visibility intervals.}
We consider the overlay of all visibility intervals of all guard candidates in~$G$, see Figure~\ref{fig:witnesses-overlay} for an overlay of three guard candidates.
It forms a subdivision consisting of \emph{maximal intervals} (maximal, connected intervals seen by the same guards) and \emph{end points.}
Every point in a \emph{feature~$f$} (maximal interval or end point) of the subdivision is seen by the same set of guards
\begin{equation}\label{eq:g-f}
	G(f) := \left\{ g \in G \mid f \subseteq \V(g) \right\}.
\end{equation}

\begin{observation}
	Let $f$ be a feature of the guard candidates' overlay, and let $g \in G$ be a guard.
	Now consider an arbitrary witness $w \in f$.
	Then
	\begin{equation}
		w \in \V(g)
			\quad \Leftrightarrow \quad f \subseteq \V(g)
			\quad \stackrel{\eqref{eq:g-f}}{\Leftrightarrow} \quad g \in G(f).
	\end{equation}
\end{observation}

\begin{observation}\label{obs:witness-representation}
	Let $G$ be a finite set of guard candidates and $f$ a feature in the overlay of~$G$.
	Then a witness point $w \in f$ can be represented by the set $G(f)$ of guards covering~$f$.
\end{observation}

Placing one witness in every feature of the subdivision ensures that coverage of all these witnesses implies coverage of all features and thus of~$T$.
This requires $\bigO(n|G|)$ witnesses.
However, keeping efficient algorithms in mind, we reduce the number of witnesses, see Section~\ref{sec:filters-witnesses}:
Similar to the shadow atomic visibility polygons in~\cite{crs-aaafmvgoag-11}\dash---a successful strategy in \ac{AGP} algorithms~\cite{rsfhkt-eag-14}\dash---it suffices to include only those features $f$ with inclusion-minimal $G(f)$, i.e., those for which no $f'$ with $G(f') \subset G(f)$ exists:

\begin{theorem}\label{thm:w}
	Consider a terrain $T$ and a finite set of guard candidates $G$ with $\V(G) = T$.
	Let $F_G$ denote the set of features of the visibility overlay of $G$ on $T$ and $w_f \in f$ an arbitrary point in the feature $f \in F_G$.
	Then for
	\begin{equation}\label{eq:w}
		W(G) := \left\{ w_f \mid f \in F_G,\text{ $G(f)$ is inclusion-minimal\,} \right\},
	\end{equation}
	we have that if $C \subseteq G$ is feasible w.r.t.\ $\tgp(G, W(G))$ then $C$ is also feasible w.r.t.\ $\tgp(G, T)$ and
	\begin{equation}\label{eq:w-opt}
		\opt(G, W(G)) = \opt(G, T).
	\end{equation}
\end{theorem}

\begin{proof}
	Let $C \subseteq G$ cover $W(G)$, and consider some point $w \in T$.
	We show that $w \in \V(C)$.
	By assumption, $w \in \V(G)$ and thus $w \in f$ for some feature $f \in F_G$.
	The set $W(G)$ contains some witness in $w_f \in f$, or a witness $w_{f'} \in f'$ with $G(f') \subseteq G(f)$ by construction.
	In the first case, $w$~must be covered, otherwise $w_f$ would not be covered and $C$ would be infeasible for $\tgp(G,W(G))$.
	In the second case $w_{f'}$ is covered, so some guard in $G(f')$ is part of $C$, and that guard also covers $f$ and therefore~$w$.

	As for Equation~\eqref{eq:w-opt}, observe that $\tgp(G, W(G))$ is a relaxation of $\tgp(G, T)$, so $\opt(G, W(G)) \leq \opt(G, T)$ follows.
	Furthermore, if $C$ is feasible and optimal w.r.t.\ $\tgp(G, W(G))$, it is also feasible for $\tgp(G, T)$ as argued above.
	It follows that $|C| = \opt(G, W(G)) \geq \opt(G, T)$, proving~\eqref{eq:w-opt}.
\end{proof}

\begin{observation}\label{obs:witcard}
	Using the set of one witness per inclusion-minimal feature as in Equation~\eqref{eq:w} may not reduce the worst-case complexity of $|W(G)| \in \bigO(n|G|)$ witnesses.
\end{observation}

\begin{proof}
	See Figure~\ref{fig:witnesses-inclusion-minimal}.
	For $|G| \in \bigTheta(n)$ consider the terrain with $\bigTheta(n)$ valleys with $\frac{|G|}{3}$ guards placed on the left (blue) and the right (red) slope each.
	In addition there is one guard (black) placed in each valley.
	Thus, each of the $\bigTheta(n)$ valleys contains $\bigTheta(|G|)$ inclusion-minimal intervals depicted in violet, resulting in $\bigO(n|G|)$ inclusion-minimal features.
\end{proof}

Nevertheless, using only inclusion-minimal witnesses significantly speeds up our implementation, refer to Sections~\ref{sec:filters-witnesses} and~\ref{sec:experiments-witnesses}.

\begin{observation}
	$W(G)$~does not require any end point $p$ between two maximal intervals $I_1$ and~$I_2$:
	$G(p) = G(I_1) \cup G(I_2)$, since visibility regions are closed sets.
\end{observation}

\subsection{Guards}
\label{sec:discretization-guards}

\begin{figure}
	\centering
	\subfigure[%
			The edge $e_i$ is critical w.r.t.\ $g_\ell$ and~$g_r$:
			The right (left) part of~$e_i$, indicated in blue (red), is seen by $g_\ell$ ($g_r$) only.]{
		\includegraphics[width=0.6\textwidth]{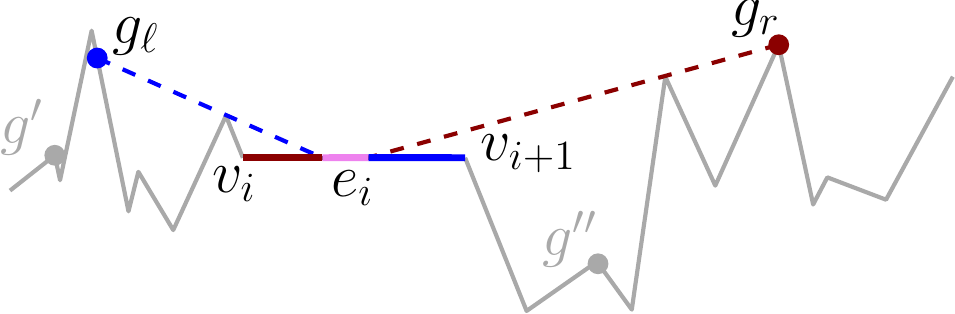}
		\label{fig:guards-critical-edge}
	}
	\subfigure[%
			No guard $g$ is both left- and right-guard.
			Any point on the critical edge $e_\ell$ seen by $g$ is also seen by~$g_r$, hence $e_\ell$ cannot be critical w.r.t.~$g$.]{
		\includegraphics[width=0.6\textwidth]{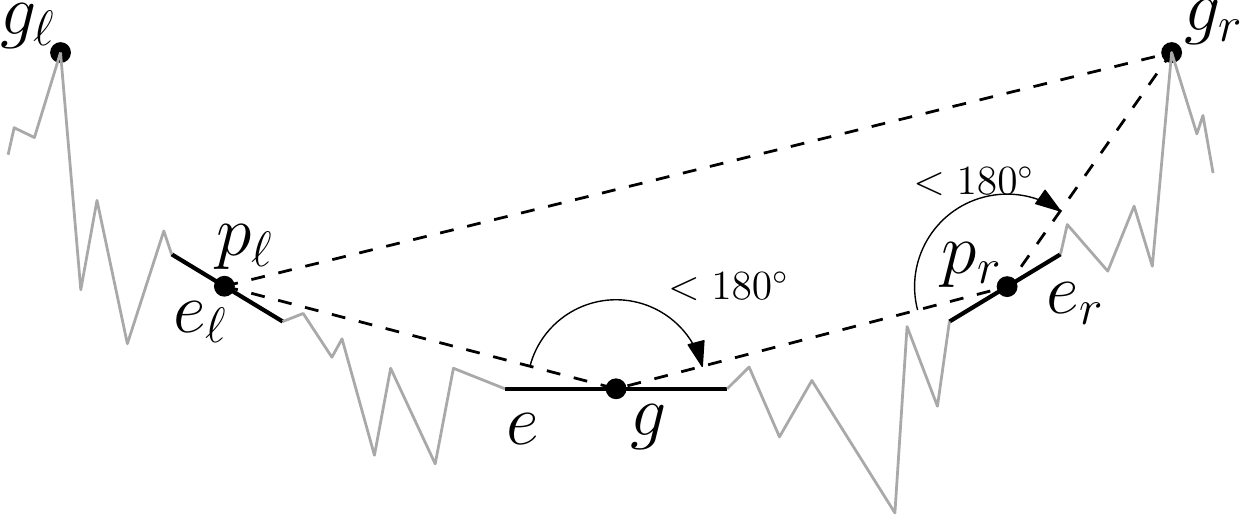}
		\label{fig:guards-leftright-guard-to-u}
	}
	\subfigure[%
			Moving the left-guard $g$ to the left.
			Any point $p$ that $g$ sees to its right remains visible while moving $g$ towards~$v_\ell$.]{
		\includegraphics[width=0.6\textwidth]{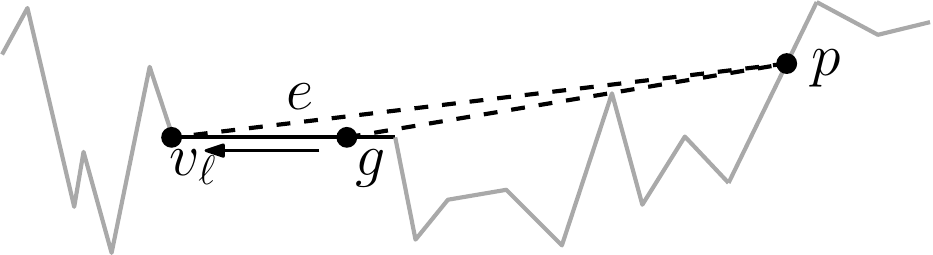}
		\label{fig:guards-left-guard-to-v}
	}
	\caption{Guard discretization: critical edges~\subref{fig:guards-critical-edge} and dominated guards~\subref{fig:guards-leftright-guard-to-u}--\subref{fig:guards-left-guard-to-v}.}
	\label{fig:guards}
\end{figure}

Throughout this section, let $T$ be a terrain, $V = V(T)$ its vertices, and $E = E(T)$ its edges.
Let $C \subset T$ be feasible w.r.t.\ $\tgp(T,T)$, i.e., some finite, possibly optimal, guard cover of~$T$.
Define $U$ as all vertices along with the extremal points of their visibility regions:
\begin{equation}\label{eq:u}
	U := V \cup \bigcup_{v \in V} \left\{ p \mid  \textnormal{$p$ is extremal in $\V(v)$} \right\}.
\end{equation}

\begin{observation}\label{obs:guardcard}
	$|U| \in \bigO(n^2)$ as noted by Ben-Moshe et~al.~\cite{bkm-acfaafotg-07}: $n$~vertices with visibility regions of $\bigO(n)$ subterrains each.
\end{observation}

Ben-Moshe et~al.\ use a similar set, but they also add an arbitrary point of $T$ between each pair of consecutive points in~$U$.
They need these points as witnesses.
We, however, keep the witnesses separate by our definition of $\tgp(G,W)$.

In the remainder of this section we show that $U$ contains all guard candidates necessary for solving the \ac{CTGP}, $\tgp(T,T)$, i.e., that $\opt(U,T) = \opt(T,T)$.
Our strategy is to show that in any cover $C$ of $T$ it is always possible to move a guard in $C \setminus U$ to a carefully chosen point in $U$ without losing coverage.
This procedure preserves the cardinality and feasibility of any feasible cover;
iterating it results in a cover $C \subseteq U$.
In particular, this is possible for an optimal guard cover.

First observe that an edge that is entirely covered by a guard $g \in C \setminus U$ is still covered after moving $g$ to one of its neighbors in~$U$.

\begin{lemma}\label{lem:guard-edge-u}
	Let $g \in C \setminus U$ be a guard that covers an entire edge $e_i \in E$.
	Then $u_\ell$ and $u_r$, the \emph{$U$-neighbors of~$g$,} with
	\begin{align}\label{eq:ul-ur}
		u_\ell & = \max\{ u \in U \mid u < g \} \\
		u_r    & = \min\{ u \in U \mid g < u \}
	\end{align}
	each entirely cover~$e_i$, too.
\end{lemma}

\begin{proof}
	$g$~covers~$e_i$, so $v_i, v_{i+1} \in \V(g)$, implying $g \in \V(v_i) \cap \V(v_{i+1})$.
	Moving $g$ towards $u_\ell$ does not move $g$ out of $\V(v_i)$ or~$\V(v_{i+1})$, as the boundaries of those regions are contained in $U$ by construction.
	Hence, $v_i, v_{i+1} \in \V(u_\ell)$ and thus $e_i \subseteq \V(u_\ell)$.
	Analogously $e_i \subseteq \V(u_r)$.
\end{proof}

It remains to consider the edges not entirely covered by a single guard, refer to Figure~\ref{fig:guards-critical-edge}.
We refer to such edges as \emph{critical edges:}

\begin{definition}[Critical Edge]\label{def:critical-edge}
	An edge $e \in E$ is \emph{critical w.r.t.\ $g \in C$} if $C \setminus \{ g \}$ covers some part of, but not all of, $\interior(e)$.
	If $e$ is critical w.r.t.\ some $g \in C$ we call $e$ \emph{critical edge.}
\end{definition}

So $e$ is critical if and only if more than one guard is responsible for covering~$\interior(e)$.

\begin{definition}[Left-Guard/Right-Guard]\label{def:left-right-guard}
	$g \in C$ is a \emph{left-guard (right-guard)} of $e_i \in E$ if $g < v_i$ ($v_{i+1} < g$) and $e_i$ is critical w.r.t.~$g$.
	We call $g$ a \emph{left-guard (right-guard)} if it is a left-guard (right-guard) of some $e \in E$.
\end{definition}

For the sake of completeness, we state and prove the following lemma which also follows from the well-established order claim~\cite{bkm-acfaafotg-07}:

\begin{lemma}\label{lem:single-interval-vis}
	Let $g \in C$ be a guard left of $v_i$ (right of $v_{i+1}$) such that $g$ covers a non-empty subset of $\interior(e_i)$.
	Then $g$ covers a single interval of~$e_i$, including $v_{i+1}$~($v_i$).
	In particular, this holds if $g$ is a left-guard (right-guard) of~$e_i$.
\end{lemma}

\begin{proof}
	Refer to Figure~\ref{fig:guards-critical-edge}.
	Obviously, $g = g_\ell$ is nowhere below the line supporting~$e_i$.
	Let $p$ be a point on $e_i$ seen by~$g_\ell$.
	It follows that $\overline{g_\ell p}$ and $\overline{p v_{i+1}}$ form an $x$-monotone convex chain that is nowhere  below~$T$.
	Thus, $\overline{g_\ell v_{i+1}}$ is nowhere below~$T$.
	It follows that $g_\ell$ sees $v_{i+1}$ and any point on $\overline{p v_{i+1}}$.
	A symmetric argument holds for the right-guard~$g_r$.
\end{proof}

\begin{corollary}\label{cor:unique-left-guard}
	For a critical edge there is exactly one left- and exactly one right-guard.
\end{corollary}

\begin{proof}
	Suppose for the sake of contradiction that $g, g' \in C$ both are critical left-guards of $e \in E$.
	By Lemma~\ref{lem:single-interval-vis}, $I := \V(g) \cap e$ and $I' := \V(g') \cap e$ are single intervals on~$e$.
	Assume w.l.o.g.\ that $I' \subseteq I$.
	This contradicts $g'$ being a left-guard of $e$ because $g$ dominates $g'$ on~$e$, i.e., $e \subseteq \V(C \setminus \{g'\})$.
	So $e$ has exactly one critical left-guard.
	A symmetric argument shows that $e$ has exactly one right-guard.
\end{proof}

\begin{corollary}\label{cor:vis-intersect}
	Let $e \in E$ be a critical edge and $g_{\ell}, g_r \in C$ be its left- and right-guards.
	Then $\V(g_\ell) \cap e \cap \V(g_r) \neq \emptyset$.
\end{corollary}

\begin{proof}
	For the sake of contradiction, suppose $I := e \setminus (\V(g_\ell) \cup \V(g_r)) \neq \emptyset$, and refer to Figure~\ref{fig:guards-critical-edge}.
	Since $C$ is feasible $I$ is covered, so some $g \in C$ sees a point $p \in I$.
	By Lemma~\ref{lem:single-interval-vis}, $g$ sees a continuous interval containing $p$ and, w.l.o.g., the right vertex of~$e$.
	It follows that $g$ dominates $g_\ell$ on $e$, contradicting that $g_\ell$ is a critical left-guard of~$e$.
\end{proof}

By Lemma~\ref{lem:guard-edge-u}, we can move non-critical guards to one of their neighbors in $U$ because they are only responsible for entire edges.
Unfortunately, this is impossible if $g \in C \setminus U$ is a left- or a right-guard:
We might lose coverage of some part of an edge that is critical w.r.t.~$g$.
However, the following lemma establishes that we can move $g$ to its left neighbor vertex if $g$ is not a right-guard (a~symmetric version for non-left-guards follows).

\begin{lemma}\label{lem:left-guard-to-u}
	Let $C$ be some finite cover of~$T$, let $g \in C \setminus V$ be a left- but not a right-guard, and let $v_\ell = \max\{ v \in V \mid v < g \}$ be the rightmost vertex left of~$g$.
	Then
	\begin{equation}
		C' = \left( C \setminus \{ g \} \right) \cup \{ v_\ell \}
	\end{equation}
	is a guard cover of~$T$.
\end{lemma}

\begin{proof}
	Since $g$ is a left-guard of some critical edge~$e_r$, there must exist a corresponding right-guard $g_r$ of $e_r$, see Figure~\ref{fig:guards-leftright-guard-to-u}.
	Let $p_\ell \in \{ p \in \V(g) \mid p \leq g \}$ be a point that $g$ sees to its left.
	We show that $p_\ell$ is seen by $g_r$:
	Consider~$p_r$, a point in $\V(g) \cap e_r \cap \V(g_r)$, which exists by Corollary~\ref{cor:vis-intersect}.
	$\overline{p_\ell g}$, $\overline{g p_r}$, and $\overline{p_r g_r}$ form a convex chain (convex due to $g, p_r \notin V$) that is nowhere below~$T$, so $p_\ell \in \V(g_r)$.
	Thus, $g$ is dominated to its left by~$g_r$.
	Moreover, $g$~is dominated to its right by $v_\ell$, see Figure~\ref{fig:guards-left-guard-to-v}:
	Let $p \in \{ p \in \V(g) \mid g \leq p \}$ be a point seen by $g$ located to its right.
	Then $\overline{v_\ell g}$ and $\overline{g p}$ form a convex chain nowhere below~$T$, so $p \in \V(v_\ell)$.
	In conclusion, replacing $g$ by $v_\ell$ in $C$ yields a feasible cover because $\{ p \in \V(g) \mid p \leq g \}$ is covered by~$g_r$ and $\{ p \in \V(g) \mid g \leq p \}$ by~$v_\ell$.
\end{proof}

\begin{corollary}\label{cor:right-guard-to-u}
	Let $C$ be some finite cover of~$T$, let $g \in C \setminus V$ be a right- but no left-guard, and let $v_r = \min\{ v \in V \mid g < v \}$ be the leftmost vertex right of~$g$.
	Then
	\begin{equation}
		C' = \left( C \setminus \{ g \} \right) \cup \{ v_r \}
	\end{equation}
	is a guard cover of~$T$.
\end{corollary}

So far, the status is that guards in $C \setminus U$ that are neither left- nor right-guard can be moved to a $U$-neighbor.
Left-guards (right-guards) that are no right-guard (left-guard) can be moved to the next vertex to the left (right).
The remaining case, i.e., guards that are both left- and right-guards, cannot happen:

\begin{lemma}\label{lem:leftright-guard-to-u}
	Let $C$ be a finite cover of~$T$.
	No $g \in C \setminus V$ is both a left- and a right-guard.
\end{lemma}

\begin{proof}
	Refer to Figure~\ref{fig:guards-leftright-guard-to-u}.
	Suppose for the sake of contradiction that $g \in C \setminus V$ is the left-guard of an edge $e_r$ (to the right of~$g$) and the right-guard of edge $e_\ell$ (to the left of~$g$).
	Since $e_r$ is critical, there must be a right-guard $g_r$ of~$e_r$.
	By Corollary~\ref{cor:vis-intersect} there is a point $p_r \in e_r$ seen by $g$ and~$g_r$.
	As $g \in C \setminus V$, $g \in \interior(e)$ for some edge~$e$.

	Now consider some point $p_\ell \in \V(g)$ such that $p_\ell < g$.
	$p_\ell$~and $p_r$ are not below the line supported by $e$ and the same holds for $g$ and $g_r$ w.r.t.~$e_r$.
	It follows that segments $\overline{p_\ell g}$, $\overline{gp_r}$, and $\overline{p_r g_r}$ form an $x$-monotone convex chain that is nowhere below $T$.
	Hence, $p_\ell \in \V(g_r)$.
	Since $p_\ell$ was arbitrary, any point $p \in \V(g)$ to the left of $g$ is also seen by~$g_r$, a contradiction to $g$ being a right-guard.
\end{proof}

The next theorem shows that the set $U$ as defined in Equation~\eqref{eq:u} contains all guard candidates necessary for a minimum-cardinality guard cover of~$T$.

\begin{theorem}\label{thm:g}
	Let $T$ be a terrain and consider $U$ from Equation~\eqref{eq:u}.
	Then we have
	\begin{equation}\label{eq:u-opt}
		\opt(U, T) = \opt(T, T).
	\end{equation}
\end{theorem}

\begin{proof}
	Let $C$ be optimal w.r.t.\ $\tgp(T,T)$.
	We show how to replace a single guard $g \in C \setminus U$ by one in~$U$ while maintaining feasibility, i.e., $\V(C) = T$.
	The claim then follows by induction.

	Should $g$ be neither left- nor right-guard, it can be replaced by a neighboring point in~$U$ by Lemma~\ref{lem:guard-edge-u}.
	If~$g$ is a left-, but not a right-guard (or vice versa), it can be replaced by its left (right) neighbor in $V \subseteq U$ by Lemma~\ref{lem:left-guard-to-u} (Corollary~\ref{cor:right-guard-to-u}).
	Lemma~\ref{lem:leftright-guard-to-u} asserts that $g$ cannot be a left- and a right-guard at the same time.
\end{proof}

\subsection{Full Discretization}
\label{sec:discretization-complete}

We formulate the key result of this section:
The \ac{CTGP}, i.e., finding a minimum-cardinality guard cover $C$ guarding an entire terrain~$T$, without any restriction on where on $T$ the guards can be placed, is a discrete problem with a discretization $(U, W(U))$ of size~$\bigO(n^3)$.

\begin{theorem}\label{thm:opt}
	Let $T$ be a terrain, and consider $U$ and $W(U)$ from Equations~\eqref{eq:u} and~\eqref{eq:w}.
	If $C \subseteq U$ is optimal w.r.t.\ $\tgp(U, W(U))$, $C$~is optimal w.r.t.\ $\tgp(T,T)$:
	\begin{equation}\label{eq:opt-equals-disc}
		\opt(T,T) = \opt(U,W(U)).
	\end{equation}
\end{theorem}

\begin{proof}
	\begin{equation}
		\opt(T, T)
			\stackrel{\eqref{eq:u-opt}}{=} \opt(U, T)
			\stackrel{\eqref{eq:w-opt}}{=} \opt(U, W(U)). \qedhere
	\end{equation}
\end{proof}

\begin{observation}\label{obs:card}
	Observations~\ref{obs:witcard} and~\ref{obs:guardcard} yield:
	The set of guard candidates $U$ and the witness set $W(U)$ have cardinality $\bigO(n^2)$ and $\bigO(n^3)$, respectively.
\end{observation}

\begin{observation}\label{obs:pol-coord}
	Let $B$ be the largest number of bits required to represent a coordinate of~$V$.
	The number of bits required to represent the coordinates of a guard candidate $g \in U$ is polynomial in $B$ as the coordinates of $g$ are defined by the intersection of two lines each spanned by two vertices in~$V$.
\end{observation}

%% file: complexity.tex
\section{Complexity Results}
\label{sec:complexity}

For a long time, the NP-hardness of the \ac{CTGP} was generally assumed, but not shown until 2010 by King and Krohn (in the conference version of~\cite{kk-tginph-11}).
In this section we establish that the \ac{CTGP} is also a member of NP, and thus NP-complete.
This is surprising, as it is a long-standing open problem for the more general \ac{AGP}:
For the \ac{AGP} it is not known whether the coordinates of an optimal guard cover can be represented with a polynomial number of bits.

\begin{theorem}
	The \acf{CTGP} is NP-complete:
	Given a terrain $T$ with rational vertices $V(T) \subset \Q^2$ and $k \in \N$, it is NP-complete to decide whether there exist $k \in \N$ guards $G = \{g_1, \dots, g_k\} \subset T$ with $\V(G) = T$.
\end{theorem}

\begin{proof}
	The NP-hardness of the \ac{CTGP} was established in~\cite{kk-tginph-11}.
	It remains to show that the \ac{CTGP} is a member of~NP:
	A non-deterministic Turing machine determines $U$ (possible in polynomial time by Observation~\ref{obs:pol-coord}), and guesses $k$ guards~$C$.
	It then verifies whether $\V(C) = T$ in polynomial time~\cite{hlmssss-tvwmv-14}.
\end{proof}

%% file: ptas.tex
\section{Polynomial Time Approximation Scheme}
\label{sec:ptas}

In this section we combine our discretization from Section~\ref{sec:discretization} with the \ac{PTAS} for discrete $\tgp(G,W)$ with finite $G,W \subset T$ by Gibson et~al.~\cite{gkkv-gtvls-14}, who established the following theorem:

\begin{theorem}[Gibson et~al.~\cite{gkkv-gtvls-14}]\label{thm:gibson}
	Let $T$ be a terrain, and let $G, W \subset T$ be finite sets of guard candidates and witnesses with $W \subseteq \V(G)$.
	Then there exists a \ac{PTAS} for $\tgp(G,W)$, i.e., for any constant $\epsilon > 0$, there is an algorithm that returns $C \subseteq G$ with $W \subseteq \V(C)$ and
	\begin{equation}\label{eq:gibson}
		|C| \leq (1 + \epsilon) \opt(G,W).
	\end{equation}
\end{theorem}

We combine our discretization from Theorem~\ref{thm:opt} with Theorem~\ref{thm:gibson}:

\begin{theorem}\label{thm:ptas}
	There is a \ac{PTAS} for the \acf{CTGP}.
	That is, for any constant $\epsilon > 0$, there is a polynomial-time algorithm which, given a terrain~$T$, returns $C \subset T$ with $\V(C) = T$ and $|C| \leq (1 + \epsilon) \opt(T,T)$.
\end{theorem}

\begin{proof}
	Using Equations~\eqref{eq:u} and~\eqref{eq:w} we determine the sets $U$ and $W(U)$ for $T$ with $|U| + |W(U)| \in \bigO(n^3)$ by Observation~\ref{obs:card}.
	By Theorem~\ref{thm:gibson}, we can compute $C \subseteq U \subset T$ with
	\begin{equation}
		|C|
			\stackrel{\eqref{eq:gibson}}{\leq} (1 + \epsilon) \opt(U, W(U))
			\stackrel{\eqref{eq:opt-equals-disc}}{=} (1 + \epsilon) \opt(T, T),
	\end{equation}
	where $C$ is feasible w.r.t.\ $\tgp(T,T)$ by Theorem~\ref{thm:opt}.
\end{proof}

%% file: filters.tex
\section{Reducing the Size of the Discretization}
\label{sec:filters}

While $\bigO(n^2)$ guard candidates and $\bigO(n^3)$ witnesses, see Observation~\ref{obs:card}, may be satisfactory from a theoretical point of view, it is imperative to reduce their numbers for an efficient implementation.
We propose filtering techniques that, while not reducing the asymptotic size of the discretization, typically remove around 90\,\% of the guard candidates and an even larger fraction of the witnesses.
Experiments in Section~\ref{sec:experiments} demonstrate this to be a key success factor that increases the solvable instance size by several orders of magnitude.

We say that $g \in T$ \emph{dominates} $g' \in T$ if $\V(g') \subseteq \V(g)$, in which case $g'$ can be safely discarded;
our filters in Sections~\ref{sec:filters-dominated} and~\ref{sec:filters-pointguards} do just that.
A core issue, however, is that visibility calculations are expensive, so the key challenge is to identify dominated guard candidates \emph{without determining their visibility region.}
The guard filter in Section~\ref{sec:filters-pointguards} has that feature.
Sections~\ref{sec:filters-witnesses} and~\ref{sec:filters-open-problems} discuss witness filtering and an open problem.

\subsection{Filtering Dominated Guards}
\label{sec:filters-dominated}

Let $T$ be a terrain and $G \subset T$ a finite set of guard candidates with $\V(G) = T$.
Consider $g, g' \in G$, suppose we know $\V(g)$ and~$\V(g')$, and observe that checking whether $g$ dominates $g'$ takes $\bigO(n)$ time since visibility regions consist of $\bigO(n)$ subterrains.
Moreover, removing all dominated guards from $G$ requires $\bigO(|G|^2)$ domination queries, i.e., an intolerable $\bigO(n^5)$ time when applied to $G = U$ from Equation~\eqref{eq:u}.

Instead, we devise a heuristic using $\bigO(|G|)$ domination queries and thus an acceptable $\bigO(n^3)$ time for $G = U$.
Suppose $G$ is ordered w.r.t.\ $x$-coordinates.
The \emph{local domination filter} removes all guard candidates that are dominated by one of their neighbors.
This is based on the observation that neighboring guards' visibility regions often are quite similar or one clearly dominates the other (a~local ``dent'').
Experiments demonstrate that this strategy is beneficial in terms of time and memory consumption, see Section~\ref{sec:experiments-dom}.

\subsection{Filtering Edge-Interior Guards}
\label{sec:filters-pointguards}

\begin{figure}
	\centering
	\includegraphics[width=\linewidth]{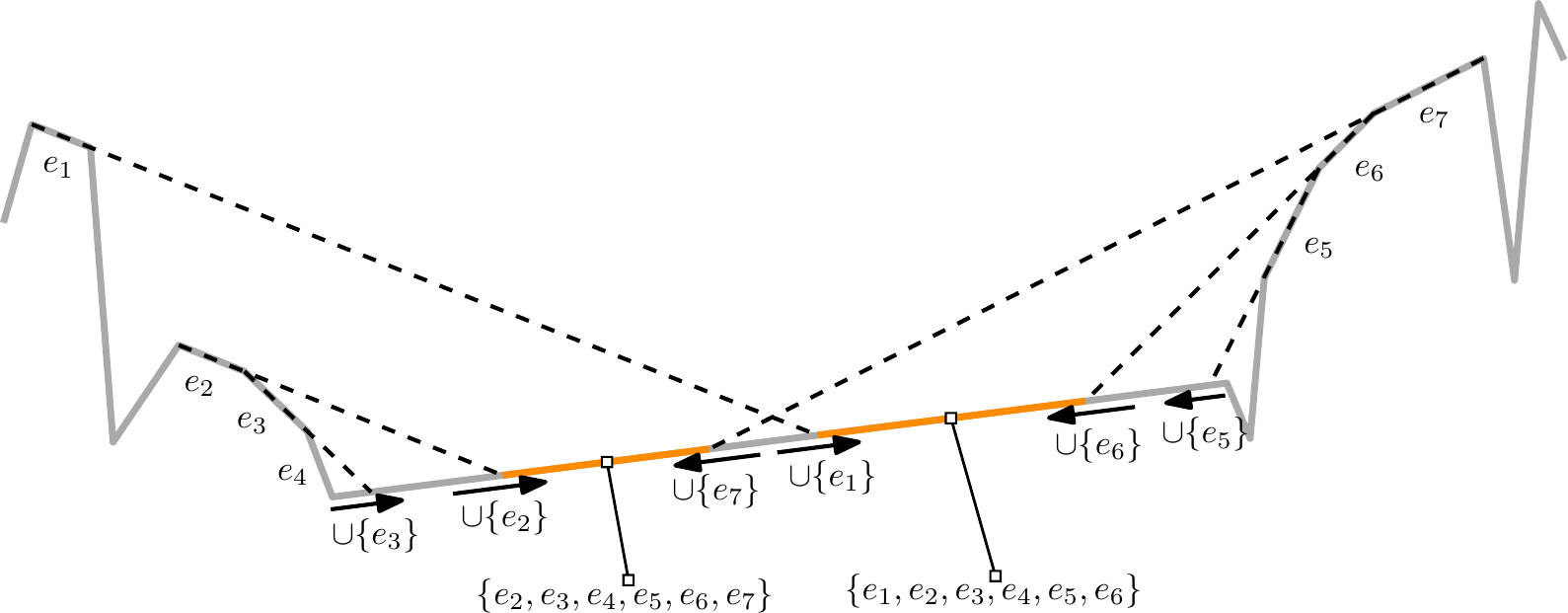}
	\caption{Edge-interior guards are only responsible for entire edges.
		Edges can only become visible when crossing some $u \in U$, the arrows indicate in which direction.
		Only the orange regions contain guard candidates that are inclusion-maximal w.r.t.\ entire edges.}
	\label{fig:point-guard-filter}
\end{figure}

Let $T$ be a terrain and $U$ the guard candidates from Equation~\eqref{eq:u} and fix an edge~$e$.
By Lemma~\ref{lem:left-guard-to-u}, Corollary~\ref{cor:right-guard-to-u}, and Lemma~\ref{lem:leftright-guard-to-u} assume w.l.o.g.\ that all critical guards are located at the vertices.
Hence, guards in $U_e := U \cap \interior(e)$ are only responsible for covering entire edges.
Recall that when moving across $u \in U_e$, a vertex becomes visible or invisible, depending on the direction, by construction of~$U$.
Furthermore, covering an entire edge is equivalent to seeing both its vertices.
The sets of edges entirely seen by each $u \in U_e$,
\begin{equation}
	E_u := \{ e \in E \mid e \subseteq \V(u) \} = \{ e_i \in E \mid v_i, v_{i+1} \in \V(u) \},
\end{equation}
define a partial ordering on $U_e$ w.r.t.\ inclusion as indicated in Figure~\ref{fig:point-guard-filter}.
Most importantly, $u$~is inclusion-maximal if $E_{u'} \not\supset E_u$ for all $u' \in U_e$.
We show that it suffices to consider the guard candidates that are inclusion-maximal w.r.t.~$E_u$:

\begin{theorem}\label{thm:guard-filter-edge}
	Let $U_e' \subseteq U_e$ be the set that only contains inclusion-maximal guard candidates w.r.t.\ entire edges, as defined above.
	Then
	\begin{equation}
		U' = (U \setminus U_e) \cup U_e'
	\end{equation}
	admits covering $T$ with the same number of guards as~$U$, i.e.,
	\begin{equation}
		\opt(U', T) = \opt(U, T).
	\end{equation}
\end{theorem}

\begin{proof}
	A guard cannot be left- and right-guard at the same time by Lemma~\ref{lem:leftright-guard-to-u}.
	Furthermore, by Lemma~\ref{lem:left-guard-to-u} (Corollary~\ref{cor:right-guard-to-u}), a left-guard (right-guard) can be moved to its left (right) neighbor in~$V$.
	Thus, w.l.o.g., $u \in U_e$ is no left- or right-guard, because $U_e$ does not contain vertices by definition.
	Hence, no edge is critical w.r.t.\ $u$ by Definition~\ref{def:critical-edge}, so $u$ is only responsible for covering entire edges and can be replaced by its inclusion-maximal sibling in $U_e'$ without changing the feasibility or cardinality of a cover of~$T$.
\end{proof}

The key is that identifying guard candidates $u \in U \setminus V$ that are not inclusion-maximal w.r.t.\ entire edges can be implemented without determining~$\V(u)$:
For each $u \in U \setminus V$, store a reference to which vertex's visibility region is extremal at~$u$, as well as whether it is situated left or right of~$u$.
This allows to decide which vertex becomes visible or invisible when sweeping across $u$ from left to right, as indicated in Figure~\ref{fig:point-guard-filter}.

We use the following sweep line algorithm.
For every $e \in E$, sweep through $U_e$ from left to right.
While encountering $u \in U_e$ where new vertices become visible, do nothing.
When reaching the first $u \in U_e$ where a vertex becomes invisible, report~$u$.
Since $u$ is inclusion-maximal w.r.t.\ vertices, it is inclusion-maximal w.r.t.\ entire edges.
Then ignore all points corresponding to vertices becoming invisible until encountering the first that becomes visible, and continue as above.

This discards up to 98\,\% of the guard candidates efficiently enough to essentially remove the computational boundary between \ac{VTGP} and \ac{CTGP}, see Section~\ref{sec:experiments-edge}.

\begin{observation}
	The above sweep line algorithm needs only the visibility regions of vertices, not those of $U \setminus V$.
	Deciding whether a vertex $v$ becomes visible or invisible at $u \in U_e$ depends only on whether $u$ is extremal in $\V(v)$ and on $v < u$, as described above.
\end{observation}

\begin{observation}\label{obs:card-filtered}
	Filtering $U$ as above still yields $\bigO(n^2)$ guard candidates:
	Insert a vertex below each guard on the slopes in Figure~\ref{fig:witnesses-inclusion-minimal}.
	Then every other interval is inclusion-maximal w.r.t.\ the vertices on the slopes.
\end{observation}

\subsection{Filtering Witnesses}
\label{sec:filters-witnesses}

Let $U$ be a possibly filtered set of guard candidates.
The construction of the witness set $W(U)$ as in Equation~\eqref{eq:w} already includes a filtering mechanism:
only inclusion-minimal witnesses need to be kept.
Observe that a smaller, filtered, $U$~automatically yields a smaller~$W(U)$.
Furthermore, observe that in terms of an implementation witnesses are much cheaper then guard candidates:
They require no visibility region or coordinates\dash---by Observation~\ref{obs:witness-representation}, they only need to store references to the guards covering them.

We acquire witnesses very much like in Section~\ref{sec:filters-pointguards}:
Sort the extremal points of all guard candidates' visibility regions by their $x$-coordinates.
For each of these points we know whether a visibility region opens or closes and to which guard it is associated.
Sweeping through these points, it is straightforward to keep track of which guard candidates see the current event point and where this set is inclusion-minimal.

Our approach keeps witnesses that are locally, but not necessarily globally, inclusion-minimal.
We can efficiently exploit the underlying geometry to identify the locally inclusion-minimal witnesses, but it is an open question whether globally inclusion-minimal witnesses can be identified just as efficiently.
However, experiments demonstrate that our approach is extremely effective, see Section~\ref{sec:experiments-witnesses}.

\subsection{Open Problem}
\label{sec:filters-open-problems}

We would like to find an optimal discretization.
But what is an optimal discretization?
Obviously, a good discretization is small\dash---asymptotically and, if an implementation is of interest, in terms of constant factors.
However, a discretization that minimizes $|G| + |W|$ is one where $|G| = \opt(T,T)$, i.e., just as hard to find as solving $\tgp(T,T)$.
Hence, we require a discretization to be obtainable in polynomial time.

Our discretization has size $|U| + |W(U)| \in \bigO(n^3)$.
The filters do not reduce the asymptotic complexity but prove effective by reducing, on average, the size by more than~90\,\%.
Is there a discretization, obtainable in polynomial time, of size~$\bigo(n^3)$?

%% file: algorithm.tex
\section{Optimal Solutions with Integer Linear Programming}
\label{sec:implementation}

We combine discretization and filters from Sections~\ref{sec:discretization} and~\ref{sec:filters} to an efficient algorithm.
It solves instances of the \ac{TGP} with up to $10^6$ vertices within roughly 1--2 minutes on a standard desktop computer, see Section~\ref{sec:experiments}.
For evaluation, the filtering techniques can be enabled individually.

\subsection{\acs{IP} Formulation}
\label{sec:ip}

Let $T$ be a terrain, and let $G, W \subset T$ be finite sets of guard candidates and witnesses, such that $W \subseteq \V(G)$.
We formulate $\tgp(G,W)$ as \ac{IP}:
\begin{alignat}{3}
	\text{min}  & \sum_{g \in G} x_g \label{eq:ip-begin} \\
	\text{s.t.} & \sum_{g \in \V(w) \cap G} x_g \geq 1 & \quad & \forall w \in W \\
	            & x_g \in \{0, 1\}                     & \quad & \forall g \in G. \label{eq:ip-end}
\end{alignat}
A binary variable $x_g$ for each guard candidate $g \in G$ indicates whether $g$ is picked:
$x_g = 1$ if and only if $g$ is part of the cover.
For each witness $w \in W$, a constraint ensures that $w$ is covered by at least one guard.
We minimize the number of guards in the cover.

Choosing $G = U$ (possibly filtered) and $W = W(U)$ from Equations~\eqref{eq:u} and~\eqref{eq:w}, \eqref{eq:ip-begin}--\eqref{eq:ip-end} model the \ac{CTGP};
picking $G = V$ and $W = W(V)$ corresponds to the \ac{VTGP}.

\subsection{Algorithm}
\label{sec:algorithm}

\begin{algorithm}[tb]
	\Input{Terrain $T$}
	\Output{Guard cover of $T$}
	\BlankLine
	$(U, W) \gets (V(T), \emptyset)$\Comment*[r]{vertices are guard candidates in both modes}
	\For{$u \in U$}{
		determine $\V(u)$\;\label{alg:visi-v}
	}
	\If(\Comment*[f]{as opposed to \vertexguardmode}){\pointguardmode}{\label{alg:ip-point-begin}
		$U \gets U \cup \bigcup_{v \in V(T)} \{ p \mid \text{$p$ is extremal in $\V(v)$} \}$\Comment*[r]{Equation~\eqref{eq:u}}
		\If{\pointguardfilter}{
			filter edge-interior guards in $U$ by sweep\Comment*[r]{Section~\ref{sec:filters-pointguards}}
		}
		\For{$u \in U \setminus V(T)$}{
			determine $\V(u)$\Comment*[r]{after \pointguardfilter, see Section~\ref{sec:filters-pointguards}} \label{alg:visi-u}
		}\label{alg:ip-point-end}
	}
	\If{\domfilter}{\label{alg:ip-dom-begin}
		filter out guards in $U$ dominated by a neighbor\Comment*[r]{Section~\ref{sec:filters-dominated}}\label{alg:ip-dom-end}
	}
	\uIf{\witnessfilter}{\label{alg:ip-witness-begin}
		$W \gets \text{inclusion-minimal features from overlay of $U$}$\Comment*[r]{Equation~\eqref{eq:w}}\label{alg:ip-witness-end}
	}
	\Else{
		$W \gets \text{all features from overlay of $U$}$\Comment*[r]{unfiltered version of Equation~\eqref{eq:w}}
	}
	solve $\tgp(U,W)$ with an \acs{IP} solver\;\label{alg:ip-solver}
	\caption{Optimal solutions for the \acs{TGP}.}
	\label{alg:ip}
\end{algorithm}

Algorithm~\ref{alg:ip} has two modes:
\pointguardmode for solving $\tgp(T,T)$, and \vertexguardmode for $\tgp(V,T)$.
Everything except lines~\ref{alg:ip-point-begin}--\ref{alg:ip-point-end} applies to both modes, lines~\ref{alg:ip-point-begin}--\ref{alg:ip-point-end} generate non-vertex guard candidates and possibly filter them.

Filtering mechanisms are activated individually:
\domfilter from Section~\ref{sec:filters-dominated} removes guard candidates that are dominated by one of their neighbors,
\pointguardfilter corresponds to the guard filter from Section~\ref{sec:filters-pointguards} and is only available in the \pointguardmode mode, and
\witnessfilter determines whether only the inclusion-minimal witness are used, refer to Equation~\eqref{eq:w}.

We remark two things about line~\ref{alg:ip-solver}.
\begin{inparaenum}
\item
	It is the only subroutine that requires exponential time, due to the NP-hardness of the \ac{TGP}~\cite{kk-tginph-11}.
	In our experiments, however, this is not the bottleneck of our algorithm;
	the geometric subroutines require most time and memory.
	We discuss this in Sections~\ref{sec:experiments-time} and~\ref{sec:experiments-memory}.
\item
	Algorithm~\ref{alg:ip} transforms an instance of \ac{VTGP} or \ac{CTGP} into an instance of \ac{SC} which it entrusts to a solver.
	An \acs{IP} or {SAT} solver, a \ac{SC} approximation algorithm, the \ac{PTAS} by Gibson et~al.~\cite{gkkv-gtvls-14}, or any other solver (including those oblivious to the underlying geometry) would work.
	Observe that all solvers benefit from our filtering framework.
	However, since benchmarking the underlying solver is not our concern, we restrict our experiments to a state-of-the-art \acs{IP} solver.
\end{inparaenum}

\subsection{Implementation}
\label{sec:implementation-frameworks}

We implemented Algorithm~\ref{alg:ip} in {C++11} and compiled with {g++}-4.8.4~\cite{gcc}.
The geometric subroutines use \acs{CGAL}-4.6~\cite{cgal} (\acl{CGAL}\acused{CGAL}) with the \texttt{CGAL::Exact\_predicates\_exact\_constructions\_kernel} kernel;
note that we follow the \ac{EGC} paradigm, i.e., use exact number types instead of floating-point arithmetic in geometric subroutines, to guarantee correctness.
Terrain visibility is solved by the implementation of Haas and Hemmer~\cite{hh-tv-15}.
We solve \acp{IP} using CPLEX-12.6.0~\cite{cplex}.
Furthermore, we use boost-1.58.0~\cite{boost} and simple-svg-1.0.0~\cite{simple-svg}.

%% file: experiments.tex
\section{Experiments}
\label{sec:experiments}

We evaluate Algorithm~\ref{alg:ip}.
It can solve large instances within minutes on a standard desktop computer;
our filtering techniques prove critical to success.
Especially \pointguardfilter and \witnessfilter, see Sections~\ref{sec:filters-pointguards} and~\ref{sec:filters-witnesses}, significantly increase the solvable instance size.
Our instances, the tested configurations of Algorithm~\ref{alg:ip}, the experimental setup, and our findings are described in Sections~\ref{sec:experiments-instance-description}, \ref{sec:configurations}, \ref{sec:experiments-setup}, and~\ref{sec:experiments-results}, respectively.

\subsection{Instances}
\label{sec:experiments-instance-description}

\begin{figure}
	\centering
	\subfigure[\walk: Random walk with uniform step width.]{
		\includegraphics[width=.95\linewidth]{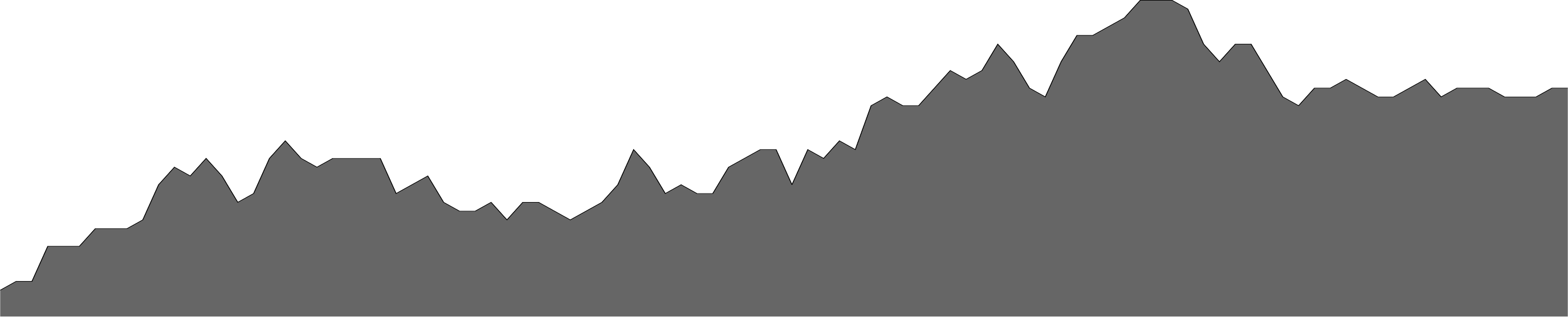}
		\label{fig:instance-walk}
	} \\
	\subfigure[\sinewalk: Sum of a sine wave and a random walk.]{
		\includegraphics[width=.95\linewidth]{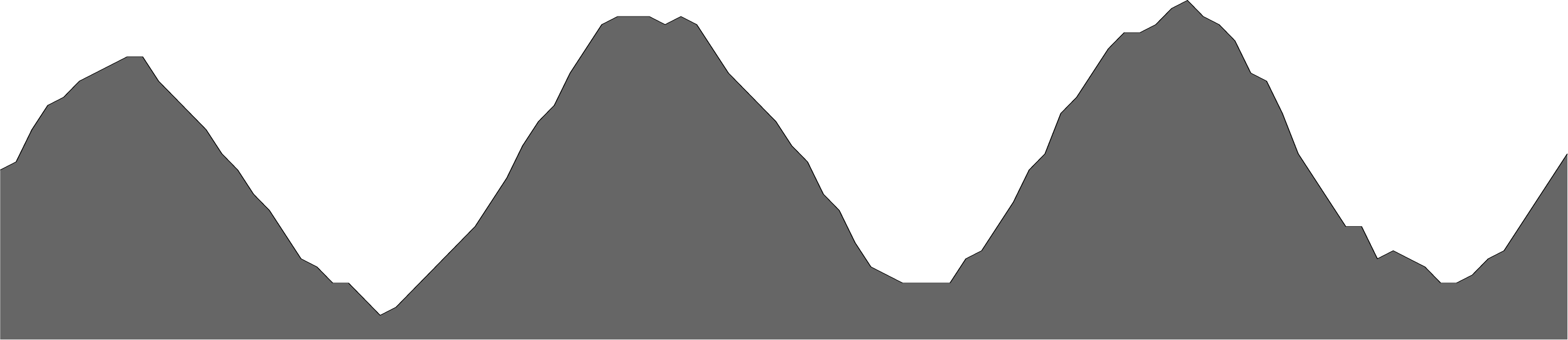}
		\label{fig:instance-sinewalk}
	} \\
	\subfigure[\parabolawalk: Sum of a pa\-rab\-o\-la and a random walk.]{
		\includegraphics[width=.45\linewidth]{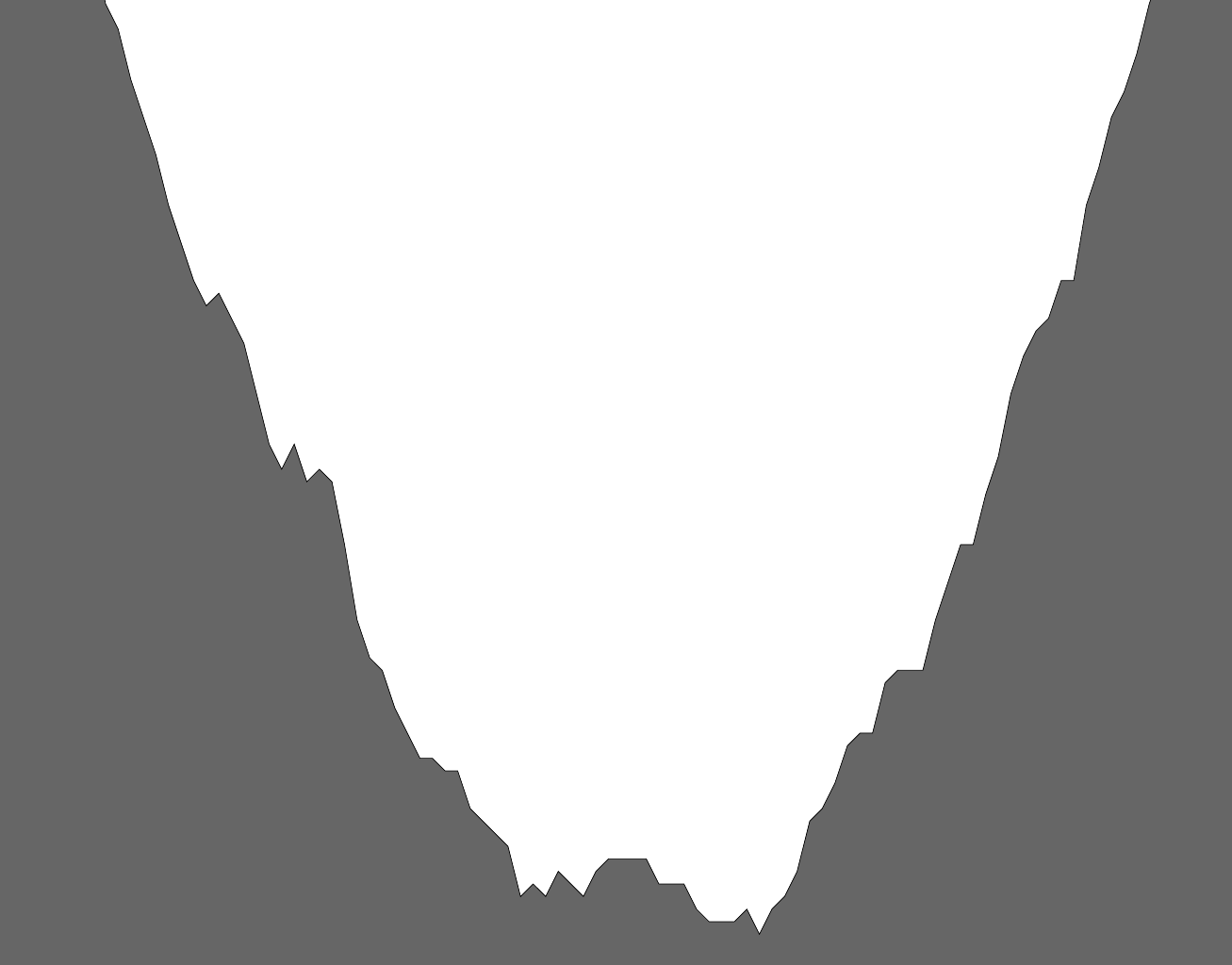}
		\label{fig:instance-parabolawalk}
	}\hspace{.03\linewidth}%
	\subfigure[\concavevalleys: Optimal solutions require point guards in the valley centers.]{
		\includegraphics[width=.45\linewidth]{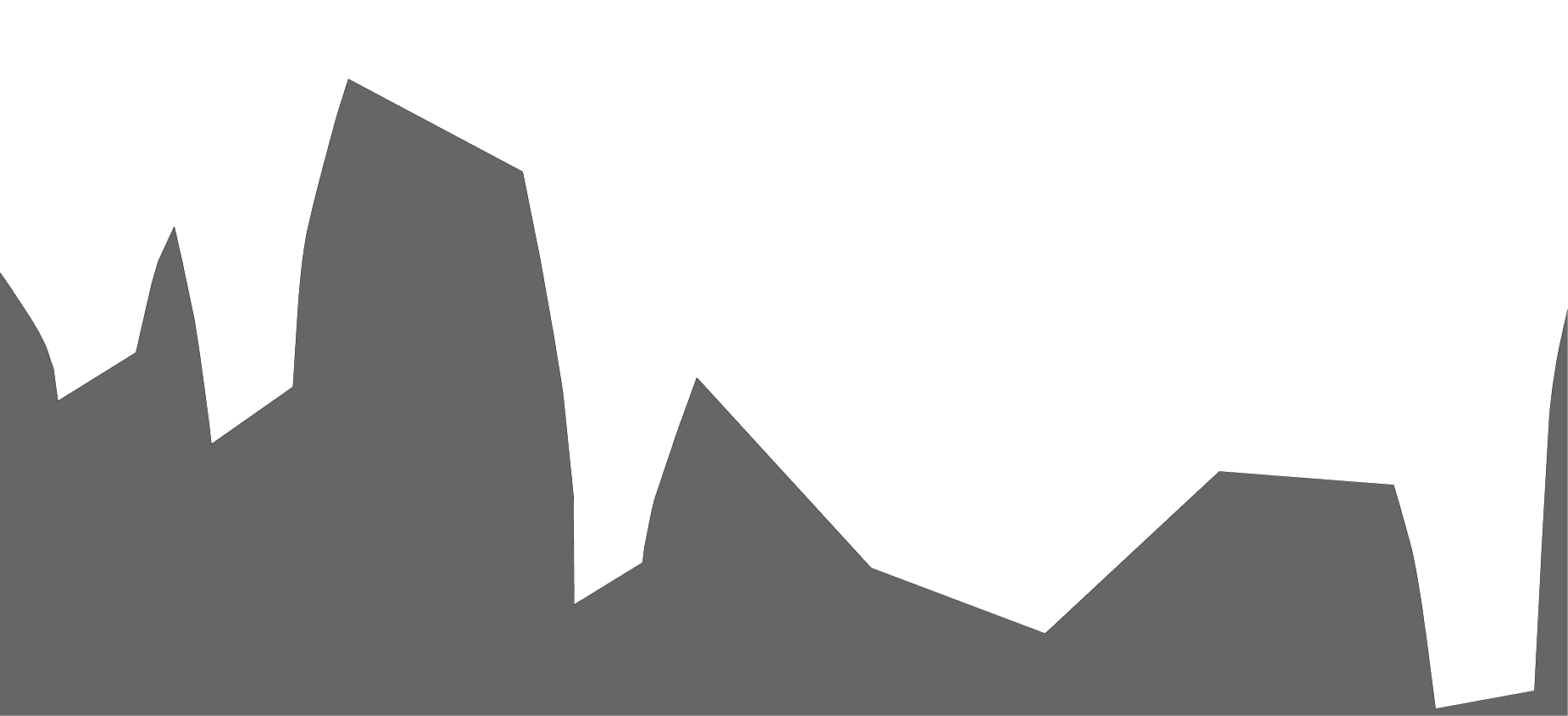}
		\label{fig:instance-concavevalleys}
	}
	\caption{Four classes of randomly generated test instances available in the \acs{TGPIL}~\cite{tgpil}.}
	\label{fig:instance}
\end{figure}

We test four classes of random terrains from the 2015-08-06 version of the \ac{TGPIL}~\cite{tgpil}, see Figure~\ref{fig:instance} for an overview.
Each class comprises 20 instances with $10^3$, $10^4$, $10^5$, $5 \cdot 10^5$, and $10^6$ vertices each, yielding 400 instances.

A \walk, see Figure~\ref{fig:instance-walk}, has $n$ vertices with $x$-coordinates $0, \dots, n-1$, and the $(i+1)$-th $y$-coordinate is a random offset from the $i$-th.
\sinewalk and \parabolawalk, see Figures~\ref{fig:instance-sinewalk} and~\ref{fig:instance-parabolawalk}, are the sum of a \walk and a properly scaled sine or parabola, respectively.
Both classes pose a challenge because many points see a large slope, highly fragmented by shadows of local features.

Preliminary experiments revealed that the above classes hardly require non-vertex guards for optimal solutions.
Hence we propose the \concavevalleys class, see Figure~\ref{fig:instance-concavevalleys}, which encourages point guards.
An instance starts as a \walk.
Iteratively pick a random edge and replace it by a valley with concave slopes.
Connect the slopes by a bottom edge, such that a point in its interior covers both slopes as in Figure~\ref{fig:tgp-point}.
An optimal solution for such a terrain usually requires guards in bottom edges' interiors.

None of the above classes deliberately provokes the NP-hardness of the \ac{TGP};
they are not designed to contain a reduction of hard instances of e.g.\ {PLANAR~3SAT}, as used in the NP-hardness proof of King and Krohn~\cite{kk-tginph-11}.
In our case, testing such instances is out of scope:
We provide and evaluate the means to transform a terrain into a small discretization that can be handed to a solver; \ac{IP}, \ac{PTAS}, {SAT}, or other.
The transformation has to be efficient and our experiments are designed to verify just that.
Combinatorially hard instances merely benchmark the underlying solver.

\subsection{Configurations}
\label{sec:configurations}

We test Algorithm~\ref{alg:ip} in seven configurations, see Table~\ref{tab:config}, to individually assess the impact of each filtering technique from Section~\ref{sec:filters}.
\vdefault, \vnodom, and \vnow test the \vertexguardmode mode;
\pdefault, \pnoedge, \pnodom, and \pnow test the considerably harder \pointguardmode mode.
Recall that \pointguardfilter does not apply to \vertexguardmode mode.
Since we test 400 instances, this results in 2800 test runs.

\begin{table}
	\centering
	\small
	\begin{tabular}{|l|c|c|c|c|}
		\hline
		Configuration & Mode & \pointguardfilter & \domfilter & \witnessfilter \\
		\hline
		\hline
		\vdefault & \vertexguardmode & n/a & yes & yes \\
		\hline
		\vnodom   & \vertexguardmode & n/a &  no & yes \\
		\hline
		\vnow     & \vertexguardmode & n/a & yes &  no \\
		\hline
		\hline
		\pdefault & \pointguardmode  & yes & yes & yes \\
		\hline
		\pnoedge  & \pointguardmode  &  no & yes & yes \\
		\hline
		\pnodom   & \pointguardmode  & yes &  no & yes \\
		\hline
		\pnow     & \pointguardmode  & yes & yes &  no \\
		\hline
	\end{tabular}
	\caption{Algorithm configurations, \acs{VTGP} above and \acs{CTGP} below.}
	\label{tab:config}
\end{table}

\subsection{Experimental Setup}
\label{sec:experiments-setup}

We used eight identical Linux 3.13 machines with Intel Core i7-3770 CPUs running at 3.4\,GHz, provided with 8\,MB of cache and 16\,GB of main memory.
Every run was limited to 15~minutes of CPU time and 14\,GB of memory.
Our software, except solving \acp{IP} with CPLEX, is not parallelized.
Refer to Section~\ref{sec:implementation-frameworks} for details regarding the toolchain.

\subsection{Results}
\label{sec:experiments-results}

The solution rates and median solution times of every combination of configuration, instance class, and instance complexity are listed in Tables~\ref{tab:percentage} and~\ref{tab:time}; each cell corresponds to 20 test runs.
Due to the imposed time and memory limits, not all test runs succeeded;
we account for unfinished test runs with an infinite completion time.
Hence, we use the median instead of the mean throughout the analysis.
More fine-grained timing information is presented in Figure~\ref{fig:time-boxplot}.
Except in \pnoedge mode, all unsolved instances were caused by running out of memory, see Section~\ref{sec:experiments-memory}.
We dedicate one subsection each to the relative hardness of the instance classes (Section~\ref{sec:experiments-instances}), an overview of \vertexguardmode and \pointguardmode modes (Sections~\ref{sec:experiments-vertex} and~\ref{sec:experiments-point}), the impact of \pointguardfilter, \domfilter and \witnessfilter (Sections~\ref{sec:experiments-edge}, \ref{sec:experiments-dom} and~\ref{sec:experiments-witnesses}), timing behavior (Section~\ref{sec:experiments-time}), and memory consumption (Section~\ref{sec:experiments-memory}).

\begin{table}
	\centering
	\small
	\begin{tabular}{|l|l|ccccc|}
		\hline
		\multirow{2}{*}{Configuration} & \multirow{2}{*}{Instance} & \multicolumn{5}{|c|}{\#vertices} \\
		& & $10^3$ & $10^4$ & $10^5$ & $5 \cdot 10^5$ & $10^6$ \\
		\hline
		\hline
		\multirow{4}{*}{\vdefault}
			& \walk           & 100\,\% & 100\,\% & 100\,\% & 100\,\% & 100\,\% \\
			& \sinewalk       & 100\,\% & 100\,\% & 100\,\% &   0\,\% &   0\,\% \\
			& \parabolawalk   & 100\,\% & 100\,\% & 100\,\% &   0\,\% &   0\,\% \\
			& \concavevalleys & 100\,\% & 100\,\% & 100\,\% & 100\,\% & 100\,\% \\
		\hline
		\multirow{4}{*}{\vnodom}
			& \walk           & 100\,\% & 100\,\% & 100\,\% & 100\,\% & 100\,\% \\
			& \sinewalk       & 100\,\% & 100\,\% & 100\,\% &   0\,\% &   0\,\% \\
			& \parabolawalk   & 100\,\% & 100\,\% & 100\,\% &   0\,\% &   0\,\% \\
			& \concavevalleys & 100\,\% & 100\,\% & 100\,\% & 100\,\% & 100\,\% \\
		\hline
		\multirow{4}{*}{\vnow}
			& \walk           & 100\,\% & 100\,\% & 100\,\% &   0\,\% &   0\,\% \\
			& \sinewalk       & 100\,\% & 100\,\% &   0\,\% &   0\,\% &   0\,\% \\
			& \parabolawalk   & 100\,\% &  25\,\% &   0\,\% &   0\,\% &   0\,\% \\
			& \concavevalleys & 100\,\% & 100\,\% & 100\,\% &   0\,\% &   0\,\% \\
		\hline
		\hline
		\multirow{4}{*}{\pdefault}
			& \walk           & 100\,\% & 100\,\% & 100\,\% & 100\,\% & 100\,\% \\
			& \sinewalk       & 100\,\% & 100\,\% & 100\,\% &   0\,\% &   0\,\% \\
			& \parabolawalk   & 100\,\% & 100\,\% & 100\,\% &   0\,\% &   0\,\% \\
			& \concavevalleys & 100\,\% & 100\,\% & 100\,\% & 100\,\% & 100\,\% \\
		\hline
		\multirow{4}{*}{\pnoedge}
			& \walk           & 100\,\% & 100\,\% & 100\,\% &  55\,\% &   0\,\% \\
			& \sinewalk       & 100\,\% & 100\,\% &   0\,\% &   0\,\% &   0\,\% \\
			& \parabolawalk   & 100\,\% & 100\,\% &   0\,\% &   0\,\% &   0\,\% \\
			& \concavevalleys & 100\,\% & 100\,\% & 100\,\% &  90\,\% &   0\,\% \\
		\hline
		\multirow{4}{*}{\pnodom}
			& \walk           & 100\,\% & 100\,\% & 100\,\% & 100\,\% & 100\,\% \\
			& \sinewalk       & 100\,\% & 100\,\% & 100\,\% &   0\,\% &   0\,\% \\
			& \parabolawalk   & 100\,\% & 100\,\% & 100\,\% &   0\,\% &   0\,\% \\
			& \concavevalleys & 100\,\% & 100\,\% & 100\,\% & 100\,\% & 100\,\% \\
		\hline
		\multirow{4}{*}{\pnow}
			& \walk           & 100\,\% & 100\,\% & 100\,\% &   0\,\% &   0\,\% \\
			& \sinewalk       & 100\,\% & 100\,\% &   0\,\% &   0\,\% &   0\,\% \\
			& \parabolawalk   & 100\,\% &  20\,\% &   0\,\% &   0\,\% &   0\,\% \\
			& \concavevalleys & 100\,\% & 100\,\% & 100\,\% &   0\,\% &   0\,\% \\
		\hline
	\end{tabular}
	\caption{Solution rates for each configuration, instance class, and instance complexity.}
	\label{tab:percentage}
\end{table}

\begin{table}
	\centering
	\small
	\begin{tabular}{|l|l|ccccc|}
		\hline
		\multirow{2}{*}{Configuration} & \multirow{2}{*}{Instance} & \multicolumn{5}{|c|}{\#vertices} \\
		& & $10^3$ & $10^4$ & $10^5$ & $5 \cdot 10^5$ & $10^6$ \\
		\hline
		\hline
		\multirow{4}{*}{\vdefault}
			& \walk           &   0.0\,s &   0.2\,s &   2.9\,s &  18.0\,s &  40.3\,s \\
			& \sinewalk       &   0.0\,s &   0.9\,s &  13.8\,s &      n/a &      n/a \\
			& \parabolawalk   &   0.1\,s &   1.8\,s &  21.7\,s &      n/a &      n/a \\
			& \concavevalleys &   0.2\,s &   2.4\,s &  24.2\,s & 177.2\,s & 337.6\,s \\
		\hline
		\multirow{4}{*}{\vnodom}
			& \walk           &   0.0\,s &   0.3\,s &   3.6\,s &  21.8\,s &  48.6\,s \\
			& \sinewalk       &   0.1\,s &   1.3\,s &  18.2\,s &      n/a &      n/a \\
			& \parabolawalk   &   0.1\,s &   2.5\,s &  27.8\,s &      n/a &      n/a \\
			& \concavevalleys &   0.2\,s &   2.4\,s &  24.3\,s & 177.1\,s & 411.0\,s \\
		\hline
		\multirow{4}{*}{\vnow}
			& \walk           &   0.0\,s &   0.4\,s &   8.8\,s &      n/a &      n/a \\
			& \sinewalk       &   0.1\,s &   6.4\,s &      n/a &      n/a &      n/a \\
			& \parabolawalk   &   0.4\,s &      n/a &      n/a &      n/a &      n/a \\
			& \concavevalleys &   0.2\,s &   2.7\,s &  32.0\,s &      n/a &      n/a \\
		\hline
		\hline
		\multirow{4}{*}{\pdefault}
			& \walk           &   0.0\,s &   0.3\,s &   4.6\,s &  27.9\,s &  62.4\,s \\
			& \sinewalk       &   0.1\,s &   1.7\,s &  26.3\,s &      n/a &      n/a \\
			& \parabolawalk   &   0.2\,s &   3.3\,s &  45.1\,s &      n/a &      n/a \\
			& \concavevalleys &   0.1\,s &   0.8\,s &  10.3\,s &  68.2\,s & 137.4\,s \\
		\hline
		\multirow{4}{*}{\pnoedge}
			& \walk           &   0.2\,s &   4.5\,s &  79.7\,s & 833.3\,s &      n/a \\
			& \sinewalk       &   1.0\,s &  58.8\,s &      n/a &      n/a &      n/a \\
			& \parabolawalk   &   4.1\,s & 220.3\,s &      n/a &      n/a &      n/a \\
			& \concavevalleys &   0.2\,s &   3.2\,s &  58.3\,s & 652.3\,s &      n/a \\
		\hline
		\multirow{4}{*}{\pnodom}
			& \walk           &   0.0\,s &   0.4\,s &   5.0\,s &  31.5\,s &  70.4\,s \\
			& \sinewalk       &   0.1\,s &   2.0\,s &  31.1\,s &      n/a &      n/a \\
			& \parabolawalk   &   0.2\,s &   4.0\,s &  51.8\,s &      n/a &      n/a \\
			& \concavevalleys &   0.1\,s &   0.9\,s &  10.7\,s &  68.0\,s & 145.5\,s \\
		\hline
		\multirow{4}{*}{\pnow}
			& \walk           &   0.0\,s &   0.6\,s &  10.4\,s &      n/a &      n/a \\
			& \sinewalk       &   0.1\,s &   7.8\,s &      n/a &      n/a &      n/a \\
			& \parabolawalk   &   0.6\,s &      n/a &      n/a &      n/a &      n/a \\
			& \concavevalleys &   0.1\,s &   1.1\,s &  20.8\,s &      n/a &      n/a \\
		\hline
	\end{tabular}
	\caption{Median solution times per configuration, instance class, and instance complexity.}
	\label{tab:time}
\end{table}

\subsubsection{Overview: Instances}
\label{sec:experiments-instances}

Tables~\ref{tab:percentage} and~\ref{tab:time} clearly reveal that \sinewalk and \parabolawalk are harder to solve than \walk and \concavevalleys.
This is to be expected, since \sinewalk and \parabolawalk contain a \walk as additive noise, and since the sole purpose of \concavevalleys is to encourage placing non-vertex guards, see Section~\ref{sec:experiments-instance-description}.
Furthermore, \sinewalk and \parabolawalk comprise facing valleys, resulting in highly fragmented visibility regions and complex visibility overlays in which a large portion of the guards and witnesses cannot be filtered out.
This induces time and memory intensive calculations and a complex \ac{IP}, making \sinewalk and \parabolawalk challenging instance classes.

\subsubsection{Overview: Vertex Guards}
\label{sec:experiments-vertex}

\vdefault and \vnodom solve all instances of \walk and \concavevalleys, and the \sinewalk and \parabolawalk instances of up to $10^5$ vertices;
\vnow can solve instances which are smaller by about a factor of~10, see Table~\ref{tab:percentage}.
This already demonstrates the importance of \witnessfilter.
\vdefault and \vnodom have comparable running times, with a slight advantage for \vdefault, see Table~\ref{tab:time}; \vnow is slower.

\subsubsection{Overview: Point Guards}
\label{sec:experiments-point}

In terms of solved instances, refer to Table~\ref{tab:percentage}, \pdefault and \pnodom are the strongest configurations, solving all \walk and \concavevalleys instances as well as the \sinewalk and \parabolawalk instances with up to $10^5$ vertices.
\pnow and \pnoedge are much weaker.
Table~\ref{tab:time} indicates that \pdefault is slightly faster than \pnodom.
It is clear from the performance of \pnoedge that \pointguardfilter is crucial in \pointguardmode mode.

\subsubsection{Impact of Filtering Edge-Interior Guards}
\label{sec:experiments-edge}

\begin{table}
	\centering
	\small
	\begin{tabular}{|l|l|ccccc|}
		\hline
		\multirow{2}{*}{Configuration} & \multirow{2}{*}{Instance} & \multicolumn{5}{|c|}{\#vertices} \\
		& & $10^3$ & $10^4$ & $10^5$ & $5 \cdot 10^5$ & $10^6$ \\
		\hline
		\hline
		\multirow{4}{*}{\pnodom}
			& \walk           & 80.1\,\% & 86.8\,\% & 89.5\,\% & 91.0\,\% & 91.7\,\% \\
			& \sinewalk       & 92.7\,\% & 97.6\,\% & 98.3\,\% &      n/a &      n/a \\
			& \parabolawalk   & 97.5\,\% & 98.8\,\% & 98.9\,\% &      n/a &      n/a \\
			& \concavevalleys & 65.8\,\% & 72.5\,\% & 77.7\,\% & 79.9\,\% & 80.6\,\% \\
		\hline
	\end{tabular}
	\caption{Median percentage of guard candidates removed by \pointguardfilter.}
	\label{tab:pointguardfilter}
\end{table}

\begin{figure}
	\subfigure[Unfiltered guard candidates.]{
		\includegraphics[width=.45\linewidth]{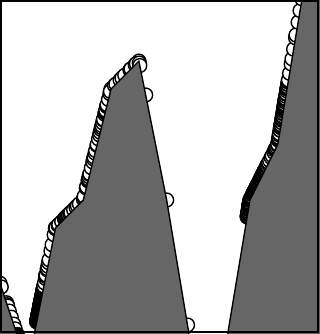}
		\label{fig:filtering-effect-off}
	}\hfill
	\subfigure[Filtered guard candidates.]{
		\includegraphics[width=.45\linewidth]{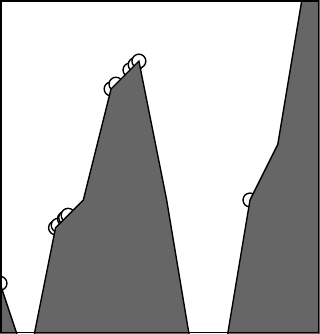}
		\label{fig:filtering-effect-on}
	}
	\caption{%
		The effect of \pointguardfilter (excerpt of a $10^5$-vertex \parabolawalk).
		White circles represent guard candidates.}
	\label{fig:filtering-effect}
\end{figure}

Recall that \pointguardfilter, see Section~\ref{sec:filters-pointguards}, only applies to \pointguardmode mode.
Table~\ref{tab:pointguardfilter} depicts the percentage of guards it removes in the \pnodom configuration\dash---the only configuration without interfering guard filters\dash---and Figure~\ref{fig:filtering-effect} illustrates its effectiveness using a $10^5$-vertex \parabolawalk as example.

\pointguardfilter proves to be our most effective guard filter by removing roughly 90\,\% (80\,\%) of the guard candidates in the $10^6$ vertex \walk (\concavevalleys) instances, and well above 95\,\% in the largest solved \sinewalk and \parabolawalk instances.
Tables~\ref{tab:percentage} and~\ref{tab:time} demonstrate that \pointguardfilter massively improves performance in terms of solution rates and median solution times.
This makes it the key success factor when solving the \ac{CTGP}, removing the computational barrier between \ac{VTGP} and \ac{CTGP}:
Without it, \pnoedge would be the state of the art, and solvable instances of the \ac{CTGP} would be smaller by at least a factor of 10 than for \ac{VTGP} and would take much more time.

\subsubsection{Impact of Filtering Dominated Guards}
\label{sec:experiments-dom}

\begin{table}
	\centering
	\small
	\begin{tabular}{|l|l|ccccc|}
		\hline
		\multirow{2}{*}{Configuration} & \multirow{2}{*}{Instance} & \multicolumn{5}{|c|}{\#vertices} \\
		& & $10^3$ & $10^4$ & $10^5$ & $5 \cdot 10^5$ & $10^6$ \\
		\hline
		\hline
		\multirow{4}{*}{\vdefault}
			& \walk           & 65.8\,\% & 64.1\,\% & 63.6\,\% & 63.5\,\% & 63.5\,\% \\
			& \sinewalk       & 58.5\,\% & 63.0\,\% & 63.6\,\% &      n/a &      n/a \\
			& \parabolawalk   & 59.9\,\% & 63.3\,\% & 63.6\,\% &      n/a &      n/a \\
			& \concavevalleys & 13.4\,\% & 13.1\,\% & 13.3\,\% & 13.3\,\% & 13.3\,\% \\
		\hline
		\hline
		\multirow{4}{*}{\pnoedge}
			& \walk           & 92.9\,\% & 94.7\,\% & 95.6\,\% & 95.8\,\% &      n/a \\
			& \sinewalk       & 88.1\,\% & 96.7\,\% &      n/a &      n/a &      n/a \\
			& \parabolawalk   & 93.2\,\% & 98.0\,\% &      n/a &      n/a &      n/a \\
			& \concavevalleys & 77.1\,\% & 80.5\,\% & 83.9\,\% & 85.0\,\% &      n/a \\
		\hline
	\end{tabular}
	\caption{Median percentage of guard candidates removed by \domfilter.}
	\label{tab:domfilter}
\end{table}

Table~\ref{tab:domfilter} displays the percentage of guard candidates that \domfilter, refer to Section~\ref{sec:filters-dominated}, filtered out in the \vdefault and \pnoedge configurations.
Observe that we need to obtain these numbers in configurations without other active filters.

\domfilter has no impact on the solution rates within the limits imposed by our setup (see Section~\ref{sec:experiments-setup}).
\vdefault and \pdefault are slightly faster than \vnodom and \pnodom, respectively.
However, the key advantage of \domfilter is that it saves memory by deleting dominated guard candidates;
this is important since memory consumption is the bottleneck of our implementation, see Section~\ref{sec:experiments-memory}.

\subsubsection{Impact of Filtering Witnesses}
\label{sec:experiments-witnesses}

\begin{table}
	\centering
	\small
	\begin{tabular}{|l|l|ccccc|}
		\hline
		\multirow{2}{*}{Configuration} & \multirow{2}{*}{Instance} & \multicolumn{5}{|c|}{\#vertices} \\
		& & $10^3$ & $10^4$ & $10^5$ & $5 \cdot 10^5$ & $10^6$ \\
		\hline
		\hline
		\multirow{4}{*}{\vdefault}
			& \walk           & 91.9\,\% & 94.3\,\% & 95.4\,\% & 96.1\,\% & 96.4\,\% \\
			& \sinewalk       & 95.8\,\% & 98.8\,\% & 99.3\,\% &      n/a &      n/a \\
			& \parabolawalk   & 98.6\,\% & 99.4\,\% & 99.5\,\% &      n/a &      n/a \\
			& \concavevalleys & 76.5\,\% & 80.5\,\% & 83.7\,\% & 85.1\,\% & 85.5\,\% \\
		\hline
		\hline
		\multirow{4}{*}{\pdefault}
			& \walk           & 91.9\,\% & 94.3\,\% & 95.5\,\% & 96.1\,\% & 96.4\,\% \\
			& \sinewalk       & 96.3\,\% & 98.9\,\% & 99.3\,\% &      n/a &      n/a \\
			& \parabolawalk   & 98.7\,\% & 99.5\,\% & 99.5\,\% &      n/a &      n/a \\
			& \concavevalleys & 80.3\,\% & 84.2\,\% & 87.4\,\% & 88.7\,\% & 89.1\,\% \\
		\hline
	\end{tabular}
	\caption{Median percentage of witnesses removed by \witnessfilter.}
	\label{tab:witnessfilter}
\end{table}

The percentage of witnesses removed by \witnessfilter, see Section~\ref{sec:filters-witnesses}, in the \vdefault and \pdefault configurations is displayed in Table~\ref{tab:witnessfilter}.
Throughout our instances, \witnessfilter removes the vast majority of witnesses, often more than~95\,\%.
Furthermore, Table~\ref{tab:percentage} clearly shows that disabling it reduces the solvable instance size by at least a factor of~10.
All of this renders \witnessfilter simple, fast, and useful.

\subsubsection{Timing Behavior}
\label{sec:experiments-time}

\begin{figure}
	\centering
	\includegraphics[width=.8\linewidth]{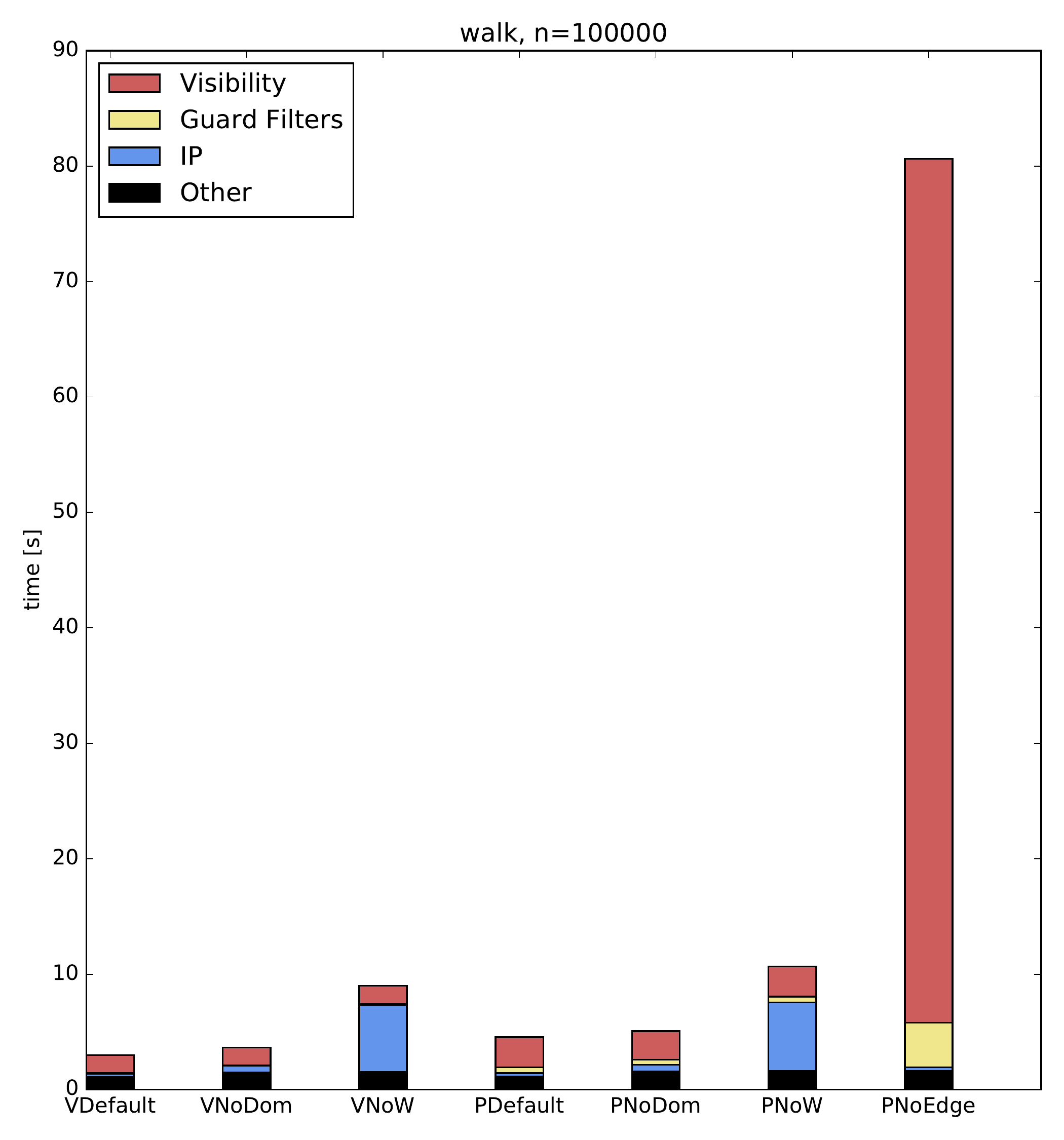}
	\caption{Median CPU time spent by subroutine regarding $10^5$-vertex \walk instances.}
	\label{fig:time-bar}
\end{figure}

\begin{figure}
	\centering
	\includegraphics[width=.8\linewidth]{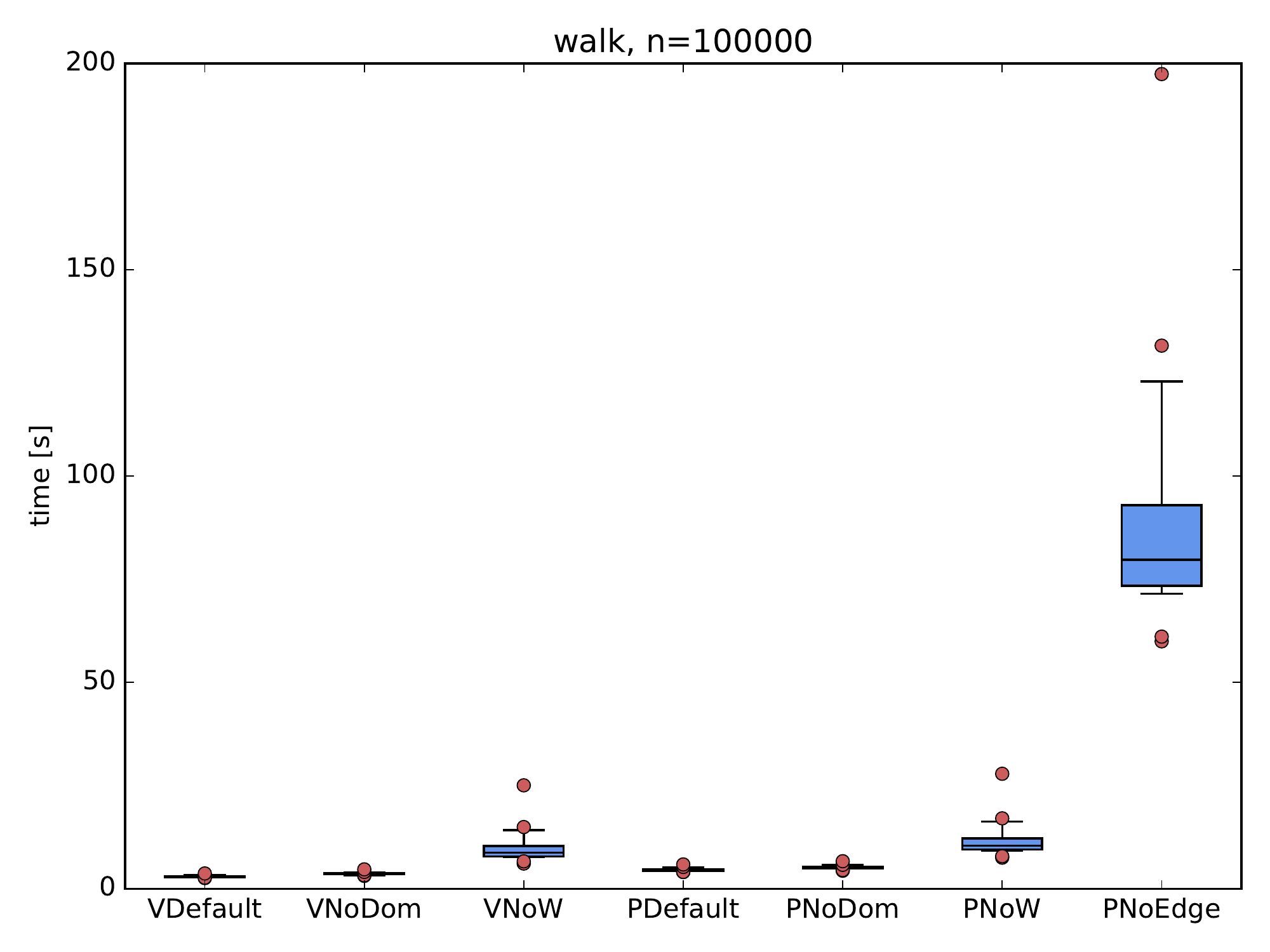}
	\caption{%
		Box plot of solution times regarding $10^5$-vertex \walk instances.
		Boxes comprise first, second, and third quartile;
		whiskers indicate the 10th and 90th percentile.}
	\label{fig:time-boxplot}
\end{figure}

Figures~\ref{fig:time-bar} and~\ref{fig:time-boxplot} show how much CPU time is spent in each part of Algorithm~\ref{alg:ip} and the distribution of the solution times, respectively.
For a comparison, we pick an instance class and a complexity that was solved by every configuration:
\walk with $10^5$ vertices\dash---the other combinations are left out, as they permit the same interpretation.

The strongest impact is that of \pointguardfilter, which is disabled in the rightmost bar in Figure~\ref{fig:time-bar}.
Without \pointguardfilter there is a computational gap between \pointguardmode and \vertexguardmode mode, as determining visibility regions of unneeded guards dominates the CPU time.
Guard-filtering time also increases in \pnoedge mode because \domfilter is active and has more guards to compare.

\witnessfilter has the second-most important impact.
In its absence, CPU times roughly double due to increased \ac{IP} solution times:
In the \ac{IP}, non-inclusion-minimal witnesses form constraints that are dominated by the inclusion-minimal witnesses' constraints.
The \ac{IP} solver eliminates dominated constraints in a preprocessing phase, but \witnessfilter does so more efficiently since it exploits the underlying geometry.

The timing behavior is hardly influenced by \domfilter.
It is, however, beneficial w.r.t.\ memory consumption, see Section~\ref{sec:experiments-memory}.

A general observation is that the filtering mechanisms significantly reduce the computational overhead.
One would expect \ac{IP} solution times to dominate an exact solver for an NP-hard problem.
This, however, tends not to be the case in geometric optimization problems like the \ac{AGP} and its relatives~\cite{rsfhkt-eag-14}.
It stems from the \ac{EGC} paradigm:
Instead of floating-point arithmetic, exact number types must be used to ensure correct results.
\vdefault and \pdefault still are not clearly dominated by \ac{IP} solution times, but much closer to it than unfiltered approaches.

\subsubsection{Memory Consumption}
\label{sec:experiments-memory}

Within our experimental setup, using the \vdefault and \pdefault modes, even instances with $10^6$ vertices are solved within minutes.
All unsolved configuration/instance pairs run out of memory, not time;
only \pnoedge occasionally runs out of time\dash---owed to the large number of unnecessary visibility calculations otherwise prevented by \pointguardfilter, see Sections~\ref{sec:experiments-edge} and~\ref{sec:experiments-time}.
So as long as instances are not designed to reveal the NP-hardness of the \ac{TGP}, the limiting resource is memory.

Two phases of Algorithm~\ref{alg:ip} generate a significant amount of data which persists in memory.
The first phase is the computation of all visibility regions of all vertices $V$ in line~\ref{alg:visi-v}.
This stores $\bigO(n^2)$ $x$-coordinates in memory which define the unfiltered guard candidate set~$U$.
Filters remove the vast majority of candidates from~$U$.
The second phase determines the visibility regions of the remaining points in $U \setminus V$, generating the largest chunk of data in line~\ref{alg:visi-u}:
At this point we keep $\bigO(n^3)$ $x$-coordinates in memory.
We conjecture that simultaneously holding all these visibility regions in memory cannot be avoided since a guard at the far right of the terrain may still see a region at its very left;
hence we do not see a way to apply \witnessfilter before knowing all extremal points.
As a lower bound, observe that a discretization with guards $G$ and witnesses $W$ yields a constraint matrix $A \in \{0,1\}^{|W| \times |G|}$ of the \acs{IP}~\eqref{eq:ip-begin}--\eqref{eq:ip-end} with $A_{wg} = 1 \Leftrightarrow w \in \V(g)$, i.e., one with $\bigO(n^5)$ entries for the \ac{CTGP}.

We remark that the memory bottleneck is amplified by the fact that we follow the \ac{EGC} paradigm, which ensures a correct and consistent representation of all visibility regions and a correct order of all visibility events.
Specifically, we do not store coordinates of points using floating-point arithmetic.
Neither do we use the other extreme, i.e., an exact representation by arbitrary precision rationals as e.g.\ provided by the GMP~\cite{gmp} library.
Instead, we rely on \ac{CGAL}~\cite{cgal}, more precisely on \texttt{CGAL::Exact\_predicates\_exact\_constructions\_kernel}, which provides lazy constructions:
Each coordinate is initially represented by two doubles that encode an interval containing the actual coordinate.
This suffices for many decisions, for instance, a comparison with another coordinate.
However, in cases in which intervals overlap, the exact coordinates are computed with GMP.
Compared to the pure exact approach this usually yields a significant advantage regarding speed and memory~\cite{pf-aglesfegc-2011}.

%% file: conclusion.tex
\section{Conclusion}
\label{sec:conclusion}

We present a discretization of polynomial size for the continuous 1.5D \ac{TGP}.
This settles two open questions:
\begin{inparaenum}
\item
	The continuous \ac{TGP} is a member of NP and, since NP-hardness is known~\cite{kk-tginph-11}, NP-complete, and
\item
	it admits a \ac{PTAS}, since the \ac{PTAS} for the discrete \ac{TGP}~\cite{gkkv-gtvls-14} applies to our discretization.
\end{inparaenum}
Furthermore, we propose an algorithm for finding optimal solutions for the \ac{TGP};
our implementation solves instances with up to $10^6$ vertices within minutes.
A key success factor are filtering techniques reducing the size of the discretization and the geometric overhead, essentially removing the computational barrier between the continuous and the discrete \ac{TGP}.